\RequirePackage{fix-cm}
\documentclass[twocolumn]{svjour3}          % twocolumn
\smartqed  % flush right qed marks, e.g. at end of proof
\usepackage{amssymb}
\usepackage{graphicx}
\usepackage{amsfonts}
\usepackage{amsmath}
\usepackage{bm}
\usepackage{hyperref}
\usepackage[linesnumbered]{algorithm2e}
\usepackage{tikz}
\usetikzlibrary{arrows, automata}
\usetikzlibrary{decorations.pathreplacing,angles,quotes}
\usepackage{xcolor} 
\usepackage{hanging}
\usepackage{bbm}
\usepackage{hyperref}
\usepackage{enumitem}
\usepackage{subcaption} 
\usepackage{changepage}
\usepackage{xcolor,colortbl}
\usepackage{multicol}% 
\usepackage{tikz}
\usetikzlibrary{arrows}
\usetikzlibrary{decorations.pathreplacing,angles,quotes}
\usepackage{tkz-tab}

\def\circarrow{{\circ\hspace{0.3mm}\!\!\! \rightarrow}}
\def\circlinecirc{{\circ \hspace{0.4mm}\!\!\! - \hspace{0.4mm}\!\!\!\circ}}
\def\circline{{\circ \! -}}
\def\arrowcirc{{\leftarrow\hspace{0.3mm}\!\!\! \circ}}

\newcommand\ci{\perp\!\!\!\perp}

\newtheorem{assumption}{Assumption}

\definecolor{amaranth}{rgb}{0.99, 0.76, 0.8}

\begin{document}

\title{Fast Causal Inference with Non-Random Missingness by Test-Wise Deletion}

\author{Eric V. Strobl \and Shyam Visweswaran \and Peter L. Spirtes
}


\institute{E. V. Strobl \at
              5607 Baum Blvd, Pittsburgh PA 15206 \\
              Tel.: 412-624-5100\\
              Fax: 412-624-5310\\
              \email{ericvonstrobl@gmail.com}          
}

\date{Received: date / Accepted: date}

\maketitle

\begin{abstract}
Many real datasets contain values missing not at random (MNAR). In this scenario, investigators often perform list-wise deletion, or delete samples with \textit{any} missing values, before applying causal discovery algorithms. List-wise deletion is a sound and general strategy when paired with algorithms such as FCI and RFCI, but the deletion procedure also eliminates otherwise good samples that contain only a few missing values. In this report, we show that we can more efficiently utilize the observed values with \textit{test-wise deletion} while still maintaining algorithmic soundness. Here, test-wise deletion refers to the process of list-wise deleting samples only among the variables required for each conditional independence (CI) test used in constraint-based searches. Test-wise deletion therefore often saves more samples than list-wise deletion for each CI test, especially when we have a sparse underlying graph. Our theoretical results show that test-wise deletion is sound under the justifiable assumption that none of the missingness mechanisms causally affect each other in the underlying causal graph. We also find that FCI and RFCI with test-wise deletion outperform their list-wise deletion and imputation counterparts on average when MNAR holds in both synthetic and real data. 
\end{abstract}

\section{The Problem}

Many real observational datasets contain missing values, but modern constraint-based causal discovery (CCD) algorithms require complete data. These facts force many investigators to either perform list-wise deletion or imputation on their datasets. The first strategy can unfortunately result in the loss of many good samples just because of a few missing values. On the other hand, the second strategy can corrupt the joint distribution when the corresponding assumptions do not hold. Both of these approaches therefore can (and often do) degrade the performance of CCD algorithms. We thus seek a practical method which allows CCD algorithms to efficiently utilize the measured values while placing few assumptions on the missingness mechanism(s). 

We specifically choose to tackle the most general case of values missing not at random (MNAR), where missing values may depend on other missing values. MNAR stands in contrast to values missing at random (MAR), where missing values can only depend on the measured values. MAR thus ensures recoverability of the underlying distribution from the measured values alone. In causal discovery, investigators usually deal with MAR by performing imputation and then running a CCD algorithm on the completed data \cite{Sokolova15,Sokolova17}. Causal discovery under MAR therefore admits a straightforward solution, once an investigator has access to a sound imputation method.

Causal discovery under MNAR requires a more sophisticated approach than causal discovery under MAR. Investigators have developed three general strategies for handling the MNAR case thus far. The first approach assumes access to some background knowledge for modeling the missingness mechanism, typically encoded using graphs \cite{Daniel12,Mohan13,Shpitser15}. Investigators with deep knowledge about the dataset at hand can therefore use this strategy to convert the MNAR problem into a more manageable form. However, access to background knowledge is arguably scarce in many situations or prone to error. The second solution involves placing an extra assumption on the missingness mechanism(s) so that we may combine the results of multiple runs of a CCD algorithm; in particular, we assume that a dataset with missing values can be decomposed into multiple datasets with potentially non-overlapping variables subject to the same set of selection variables \cite{Tillman08,Tillman11,Tillman14,Triantafilou10}. The problem of missing values therefore reduces to a problem of combining multiple datasets. Investigators nevertheless often find the assumption of identical selection bias across datasets hard to justify in practice. The third most general solution involves running a CCD algorithm that can handle selection bias on a list-wise deleted dataset, where investigators remove samples that contain any missing values \cite{Spirtes01}. List-wise deletion is nonetheless sample inefficient, because it eliminates samples with only a mild number of missing values. We therefore conclude that the three aforementioned strategies for the MNAR case can carry unsatisfactory limitations in real situations.

In this report, we propose to handle the MNAR case in CCD algorithms using a different strategy involving test-wise deletion. Here, test-wise deletion refers to the process of only performing list-wise deletion among the variables required for each conditional independence (CI) test. We develop the test-wise deletion procedure in detail throughout this report as follows. First, we provide background material in Section \ref{sec_prelim}. We then characterize missingness using graphical models augmented with missingness indicators in Sections \ref{sec_SB} and \ref{sec_missing}. Next, we justify test-wise deletion in Sections \ref{sec_ass} and \ref{sec_theory} under the assumption that certain sets of missingness indicators do not causally affect each other in the underlying causal graph. These results lead to our final solution in Section \ref{sec_alg}. We also list experimental results in Section \ref{sec_exps} which highlight the benefits of the Fast Causal Inference (FCI) algorithm and the Really Fast Causal Inference (RFCI) algorithm with test-wise deletion as opposed to the same algorithms with list-wise deletion or imputation. Finally, we conclude the paper with a short discussion in Section \ref{sec_conc}.

\section{Preliminaries} \label{sec_prelim}

We will represent causality by Markovian graphs. We therefore require some basic graphical definitions.

A graph $\mathbb{G}=(\bm{X}, \mathcal{E})$ consists of a set of vertices $\bm{X}=\{ X_1, \dots, X_p \}$ and a set of edges $\mathcal{E}$. The edge set $\mathcal{E}$ may contain the following six edge types: $\rightarrow$ (directed), $\leftrightarrow$ (bidirected), --- (undirected), $\circarrow$ (partially directed), $\circline$ (partially undirected) and $\circlinecirc$ (nondirected). Notice that these six edges utilize three types of endpoints including \textit{tails}, \textit{arrowheads}, and \textit{circles}.

We call a graph containing only directed edges as a \textit{directed graph}. On the other hand, a \textit{mixed graph} contains directed, bidirected and undirected edges. We say that $X_i$ and $X_j$ are \textit{adjacent} in a graph, if they are connected by an edge independent of the edge's type. An \textit{(undirected) path} $\pi$ between $X_i$ and $X_j$ is a set of consecutive edges (also independent of their type) connecting the variables such that no vertex is visited more than once. A \textit{directed path} from $X_i$ to $X_j$ is a set of consecutive directed edges from $X_i$ to $X_j$ in the direction of the arrowheads. A \textit{cycle} occurs when a path exists from $X_i$ to $X_j$, and $X_j$ and $X_i$ are adjacent. More specifically, a directed path from $X_i$ to $X_j$ forms a \textit{directed cycle} with the directed edge $X_j \rightarrow X_i$ and an \textit{almost directed cycle} with the bidirected edge $X_j \leftrightarrow X_i$.

Three vertices $\{X_i,X_j,X_k\}$ form an \textit{unshielded triple}, if $X_i$ and $X_j$ are adjacent, $X_j$ and $X_k$ are adjacent, but $X_i$ and $X_k$ are not adjacent. We call a nonendpoint vertex $X_j$ on a path $\pi$ a \textit{collider} on $\pi$, if both the edges immediately preceding and succeeding the vertex have an arrowhead at $X_j$. Likewise, we refer to a nonendpoint vertex $X_j$ on $\pi$ which is not a collider as a \textit{non-collider}. Finally, an unshielded triple involving $\{X_i,X_j,X_k\}$ is more specifically called a \textit{v-structure}, if $X_j$ is a collider on the subpath $\langle X_i,X_j,X_k \rangle$.

We say that $X_i$ is an \textit{ancestor} of $X_j$ (and $X_j$ is a \textit{descendant} of $X_i$) if and only if there exists a directed path from $X_i$ to $X_j$ or $X_i = X_j$. We write $X_i \in \bm{An}(X_j)$ to mean $X_i$ is an ancestor of $X_j$ and $X_j \in \bm{De}(X_i)$ to mean $X_j$ is a descendant of $X_i$. We also apply the definitions of an ancestor and descendant to a set of vertices $\bm{Y} \subseteq \bm{X}$ as follows: 
\begin{equation} \nonumber
\begin{aligned}
\bm{An}(\bm{Y}) &= \{X_i | X_i \in \bm{An}(X_j) \text{ for some } X_j \in \bm{Y}\},\\
\bm{De}(\bm{Y}) &= \{X_i | X_i \in \bm{De}(X_j) \text{ for some } X_j \in \bm{Y}\}.
\end{aligned}
\end{equation}
We call a directed graph a \textit{directed acyclic graph} (DAG), if it does not contain directed cycles. Every DAG is a type of \textit{ancestral graph}, or a mixed graph that (1) does not contain directed cycles, (2) does not contain almost directed cycles, and (3) for any undirected edge $X_i - X_j$
in $\mathcal{E}$, $X_i$ and $X_j$ have no parents or spouses.

\subsection{Causal Interpretation of Graphs} \label{sec_prob_graphs}

Consider a stochastic causal process with a distribution $\mathbb{P}$ over $\bm{X}$ that satisfies the \textit{Markov property}. A distribution satisfies the Markov property if it admits a density that ``factorizes according to the DAG'' as follows:
\begin{equation} \label{fac}
f(\bm{X})=\prod_{i=1}^{p} f(X_i | \bm{Pa}(X_i)).
\end{equation}
\noindent We can in turn relate \eqref{fac} to a graphical criterion called d-connection. Specifically, if $\mathbb{G}$ is a directed graph in which $\bm{A}$, $\bm{B}$ and $\bm{C}$ are disjoint sets of vertices in $\bm{X}$, then $\bm{A}$ and $\bm{B}$ are \textit{d-connected} by $\bm{C}$ in the directed graph $\mathbb{G}$ if and only if there exists an \textit{active path} $\pi$ between some vertex in $\bm{A}$ and some vertex in $\bm{B}$ given $\bm{C}$. An active path between $\bm{A}$ and $\bm{B}$ given $\bm{C}$ refers to an undirected path $\pi$ between some vertex in $\bm{A}$ and some vertex in $\bm{B}$ such that, for every collider $X_i$ on $\pi$, a descendant of $X_i$ is in $\bm{C}$ and no non-collider on $\pi$ is in $\bm{C}$. A path is \textit{inactive} when it is not active. Now $\bm{A}$ and $\bm{B}$ are \textit{d-separated} by $\bm{C}$ in $\mathbb{G}$ if and only if they are not d-connected by $\bm{C}$ in $\mathbb{G}$. For shorthand, we will write $\bm{A} \ci_d \bm{B} | \bm{C}$ and $\bm{A} \not \ci_d \bm{B} | \bm{C}$ when $\bm{A}$ and $\bm{B}$ are d-separated or d-connected given $\bm{C}$, respectively. The conditioning set $\bm{C}$ is called a \textit{minimal separating set} if and only if $\bm{A} \ci_d \bm{B} | \bm{C}$ but $\bm{A}$ and $\bm{B}$ are d-connected given any proper subset of $\bm{C}$. 

Now if we have $\bm{A} \ci_d \bm{B} | \bm{C}$, then $\bm{A}$ and $\bm{B}$ are conditionally independent given $\bm{C}$, denoted as $\bm{A} \ci \bm{B} | \bm{C}$, in any joint density factorizing according to \eqref{fac}; we refer to this property as the \textit{global directed Markov property}. We also refer to the converse of the global directed Markov property as \textit{d-separation faithfulness}; that is, if $\bm{A} \ci \bm{B}|\bm{C}$, then $\bm{A}$ and $\bm{B}$ are d-separated given $\bm{C}$. One can in fact show that the factorization in \eqref{fac} and the global directed Markov property are equivalent, so long as the distribution over $\bm{X}$ admits a density \cite{Lauritzen90}.\footnote{We will only consider distributions which admit densities in this report.}

 A \textit{maximal ancestral graph} (MAG) is an ancestral graph where every missing edge corresponds to a conditional independence relation. We specifically partition $\bm{X} = \bm{O} \cup \bm{L} \cup \bm{S}$ into observable, latent and selection variables, respectively. One can then transform a DAG $\mathbb{G}=(\bm{O} \cup \bm{L} \cup \bm{S}, \mathcal{E})$ into a MAG $\widetilde{\mathbb{G}}=(\bm{O},\widetilde{\mathcal{E}})$ as follows. First, for any pair of vertices $\{O_i, O_j\}$, make them adjacent in $\widetilde{\mathbb{G}}$ if and only if there is an \textit{inducing path} between $O_i$ and $O_j$ in $\mathbb{G}$. We define an inducing path as follows:
\begin{definition}
A path $\pi$ between $O_i$ and $O_j$ is called an inducing path with respect to $\bm{L}$ and $\bm{S}$ if and only if every collider on $\pi$ is an ancestor of $\{O_i,O_j\} \cup \bm{S}$, and every non-collider on $\pi$ (except for the endpoints) is in $\bm{L}$.
\end{definition}
\noindent Note that two observables $O_i$ and $O_j$ are connected by an inducing path if and only if there are d-connected given any $\bm{W} \subseteq \bm{O} \setminus \{ O_i, O_j \}$ as well as $\bm{S}$. Then, for each adjacency $O_i *\!\! -\!\! * O_j$ in $\widetilde{\mathbb{G}}$, place an arrowhead at $O_i$ if $O_i \not \in \bm{An}(O_j \cup \bm{S})$ and place a tail otherwise. The MAG of a DAG is therefore a kind of marginal graph that does not contain the latent or selection variables, but does contain information about the ancestral relations between the observable and selection variables in the DAG. The MAG also has the same d-separation relations as the DAG, specifically among the observable variables conditional on the selection variables \cite{Spirtes96}. 

\subsection{The FCI Algorithm} \label{sec_FCI}

The FCI algorithm considers the following problem: assume that the distribution of $\bm{X} = \bm{O} \cup \bm{L} \cup \bm{S}$ is d-separation faithful to an unknown DAG. Then, given oracle information about the conditional independencies between any pair of variables $O_i$ and $O_j$ given any $\bm{W} \subseteq \bm{O}\setminus \{O_i,O_j \}$, reconstruct as much information about the underlying DAG as possible. The FCI algorithm ultimately accomplishes this goal by reconstructing a MAG up to its Markov equivalence class, or the set of the MAGs with the same conditional independence relations over $\bm{O}$ given $\bm{S}$.

The FCI algorithm represents the Markov equivalence class of MAGs, or the set of MAGs with the same conditional dependence and independence relations between variables in $\bm{O}$ given $\bm{S}$, using a \textit{completed partial maximal ancestral graph} (CPMAG).\footnote{The CPMAG is also known as a partial ancestral graph (PAG). However, we will use the term CPMAG in order to mimic the use of the term CPDAG.} A \textit{partial maximal ancestral graph} (PMAG) is nothing more than a MAG with some circle endpoints. A PMAG is \textit{completed} when the following conditions hold: (1) every tail and arrowhead also exists in every MAG belonging to the Markov equivalence class of the MAG, and (2) there exists a MAG with a tail and a MAG with an arrowhead in the Markov equivalence class for every circle endpoint. Each edge in the CPMAG also has the following interpretation:
\begin{enumerate}[label=(\roman*)]\label{edge_interp_1}
\item An edge is absent between two vertices $O_i$ and $O_j$ if and only if there exists some $\bm{W} \subseteq \bm{O}\setminus \{O_i, O_j\}$ such that $O_i \ci O_j | (\bm{W}, \bm{S})$. That is, an edge is absent if and only if there does not exist an inducing path between $O_i$ and $O_j$ given $\bm{S}$.
\item If an edge between $O_i$ and $O_j$ has an arrowhead at $O_j$, then $O_j \not \in \bm{An}(O_i \cup \bm{S})$.
\item If an edge between $O_i$ and $O_j$ has a tail at $O_j$, then $O_j \in \bm{An}(O_i \cup \bm{S})$.
\end{enumerate}

The FCI algorithm learns the CPMAG through a three step procedure. Most of the algorithmic details are not important for this paper, so we refer the reader to \cite{Spirtes00} and \cite{Zhang08} for algorithmic details. However, three components of FCI called \textit{v-structure discovery}, \textit{orientation rule 1} (R1), and the \textit{discriminating path rule} (R4) are important. V-structure discovery reads as follows: suppose $O_i$ and $O_k$ are adjacent, $O_j$ and $O_k$ are adjacent, but $O_i$ and $O_j$ are non-adjacent. Further assume that we have $O_i \ci O_j | (\bm{W},$ $\bm{S})$ with $\bm{W} \subseteq \bm{O} \setminus \{O_i, O_j\}$ and $O_k \not \in \bm{W}$. Then orient the triple $\langle O_i, O_k, $ $O_j \rangle$ as $O_i \circarrow O_k \arrowcirc O_j$. R1 reads as follows: if we have $O_i \circarrow O_k \circline \!\! * O_j$, $O_i$ $\ci O_j | (\bm{W}, \bm{S})$ with $\bm{W} \subseteq \bm{O} \setminus \{O_i, O_j\}$ minimal, and $O_k \in \bm{W}$, then orient $O_i \circarrow O_k \circline \!\! * O_j$ as $O_i \circarrow O_k \rightarrow O_j$; here, the asterisk represents a placeholder for either a tail, arrowhead or circle. R4 involves the detection of additional colliders in certain shielded triples. 

\subsection{The RFCI Algorithm}

Discovering inducing paths can require large d-separating sets, so the FCI algorithm often takes too long to complete. The RFCI algorithm \cite{Colombo12} resolves this problem by recovering a graph where the presence and absence of an edge have the following modified interpretations:
\begin{enumerate}[label=(\roman*)]\label{edge_interp_2}
\item The absence of an edge between two vertices $O_i$ and $O_j$ implies that there exists some $\bm{W} \subseteq \bm{O}\setminus \{O_i, O_j\}$ such that $O_i \ci O_j | (\bm{W}, \bm{S})$.
\item The presence of an edge between two vertices $O_i$ and $O_j$ implies that $O_i \not \ci O_j | (\bm{W}, \bm{S})$ for all $\bm{W} \subseteq \bm{Adj}(O_i) \setminus O_j$ and for all $\bm{W} \subseteq \bm{Adj}(O_j) \setminus O_i$. Here $\bm{Adj}(O_i)$ denotes the set of vertices adjacent to $O_i$ in RFCI's graph.
\end{enumerate}
\noindent We encourage the reader to compare these edge interpretations to the edge interpretations of FCI's CPMAG. 

The RFCI algorithm learns its graph (not necessarily a CPMAG) also through a three step procedure. We refer the reader to \cite{Colombo12} for algorithmic details.

\section{Selection Bias} \label{sec_SB}

\textit{Selection bias} refers to the preferential selection of samples from $\mathbb{P}_{\bm{X}}$, potentially due to some unknown factors $\bm{S} \subseteq \bm{X}$. Such preferential selection occurs in a variety of real-world contexts. For example, a psychologist may wish to discover principles of the mind that apply to the general population, but he or she may only have access to data collected from college students. A medical investigator may similarly wish to elucidate a disease process occurring in all patients with the disease, but he or she may only have samples collected from low income patients in Chicago who chose to enroll in the investigator's study.

We can represent selection bias graphically using a DAG over $\bm{X}= \{\bm{O} \cup \bm{L} \cup \bm{S} \}$. We specifically let $\bm{S}$ denote a set of binary indicator variables taking values in $\{0,1\}$. Wlog, we then say that a sample is selected if and only if all of the indicator variables in $\bm{S}$ take on a value of one. The preferential selection of samples due to selection bias therefore amounts to conditioning on $\bm{S}=1$; in other words, we no longer have access to i.i.d. samples from $\mathbb{P}_{\bm{O}}$ but rather i.i.d samples from $\mathbb{P}_{\bm{O}|\bm{S}=1}$.

As an example, let $\bm{X}=\{X_1, \dots, X_5\}$, $\bm{O} = \{ X_1, X_3 \}$ and $\bm{L}=\{X_4, X_5\}$. Also let $\bm{S} = X_2$ correspond to a binary variable taking the value of 1 when $X_1$ is less than 50K and 0 otherwise. Consider drawing i.i.d. samples from a joint distribution $\mathbb{P}_{\bm{X}}$ as shown in Figure \ref{table_SB1}; here, each \textit{sample} corresponds to a row in the table. The caveat however is that we can only observe the values of $\{ X_1, X_3 \}$ when $X_1$ is below 50K as highlighted in blue in Figure \ref{table_SB1}. We therefore observe $\{X_1, X_3\}$ when $X_2 = \bm{S}$ takes on a value of 1, and otherwise we do not. In the real world, this situation may correspond to a physician who wants to measure the income $X_1$ and resting systolic blood pressure (SBP) $X_3$ of patients in the true patient population. The physician can nevertheless only measure $\{X_1, X_3\}$ in patients with low income, since patients with low income tend to enroll in medical studies more often than patients with high income. Thus, we no longer have access to i.i.d. samples from $\mathbb{P}_{X_1 X_3}$ but rather i.i.d. samples from $\mathbb{P}_{X_1 X_3|X_2=1}$, or equivalently $\mathbb{P}_{\bm{O}| \bm{S}=1}$.

We can represent the causal process in the above example using the probabilistic DAG represented in Figure \ref{fig_SB}. Here, we interpret $X_2$ as a child of $X_1$, since $X_2$ represents an indicator variable that takes on values according to the values of $X_1$. Notice also the double sided vertex in Figure \ref{fig_SB} which denotes the conditioning on low income when $X_2=1$.

\definecolor{babyblue}{rgb}{0.54, 0.81, 0.94}
\definecolor{celadon}{rgb}{0.67, 0.88, 0.69}

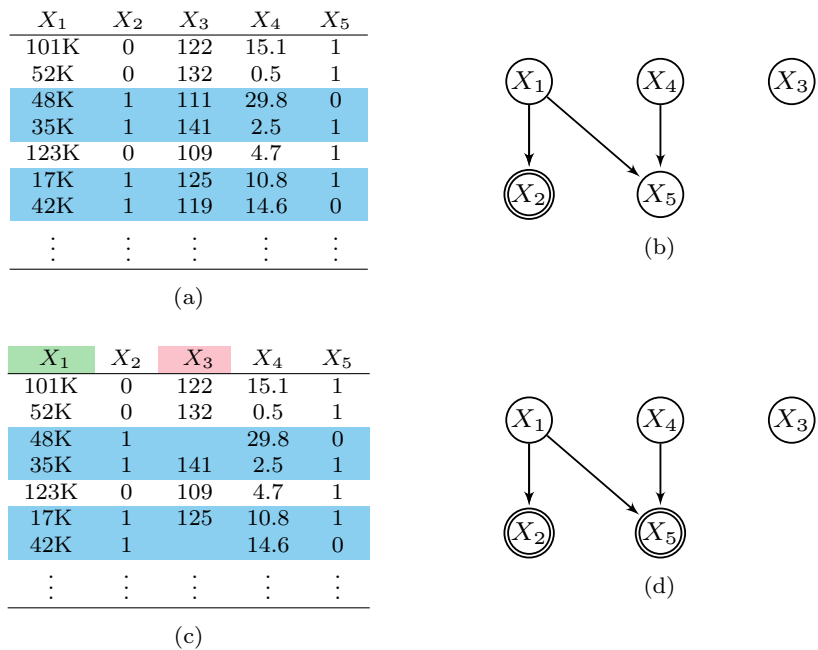
\begin{figure*}
\centering
\begin{subtable}{0.45\textwidth}
  \centering
\begin{tabular}{ c c c c c}
  $X_1$ & $X_2$ & $X_3$ & $X_4$ & $X_5$\\
	\hline
  101K & 0 & 122 & 15.1 & 1\\
  52K & 0 & 132 & 0.5 & 1\\
  \rowcolor{babyblue}
  48K & 1 & 111& 29.8 & 0\\
  \rowcolor{babyblue}
  35K & 1 & 141 & 2.5 & 1\\
  123K & 0 & 109 & 4.7 & 1\\
  \rowcolor{babyblue}
  17K & 1 & 125 & 10.8 & 1\\
  \rowcolor{babyblue}
  42K & 1 & 119 & 14.6 & 0\\ 
  $\vdots$ & $\vdots$ & $\vdots$ & $\vdots$ & $\vdots$\\
  \hline
\end{tabular}
\caption{} \label{table_SB1}
\end{subtable}
\begin{subfigure}{.25\linewidth}
\centering
\resizebox{\linewidth}{!}{
\begin{tikzpicture}[scale=1.0, shorten >=1pt,auto,node distance=2.8cm, semithick]
                    
\tikzset{vertex/.style = {shape=circle,draw,inner sep=0.4pt}}
\tikzset{edge/.style = {->,> = latex'}}
 
\node[vertex] (1) at  (0,1) {$X_1$};
\node[vertex,accepting] (2) at  (0,-0.3) {$X_2$};
\node[vertex] (4) at  (1.5,1) {$X_4$};
\node[vertex] (5) at  (1.5,-0.3) {$X_5$};
\node[vertex] (3) at  (3,1) {$X_3$};

\draw[edge] (1) to (2);
\draw[edge] (1) to (2);
\draw[edge] (4) to (5);
\draw[edge] (1) to (5);
\end{tikzpicture}
}
\caption{}  \label{fig_SB}
\end{subfigure}
\par\bigskip
\begin{subtable}{0.45\textwidth}
  \centering

\begin{tabular}{c c c c c}
  \cellcolor{celadon} $X_1$ & $X_2$ & \cellcolor{amaranth} $X_3$ & $X_4$ & $X_5$\\
	\hline
  101K & 0 & 122 & 15.1 & 1\\
  52K & 0 & 132 & 0.5 & 1\\
  \rowcolor{babyblue}
  48K & 1 &  & 29.8 & 0\\
  \rowcolor{babyblue}
  35K & 1 & 141 & 2.5 & 1\\
  123K & 0 & 109 & 4.7 & 1\\
  \rowcolor{babyblue}
  17K & 1 & 125 & 10.8 & 1\\
  \rowcolor{babyblue}
  42K & 1 &  & 14.6 & 0\\ 
  $\vdots$ & $\vdots$ & $\vdots$ & $\vdots$ & $\vdots$\\
  \hline 
\end{tabular}
\caption{} \label{table_SB2}
\end{subtable}
\begin{subfigure}{.25\linewidth}
\centering
\resizebox{\linewidth}{!}{
\begin{tikzpicture}[scale=1.0, shorten >=1pt,auto,node distance=2.8cm, semithick]
                    
\tikzset{vertex/.style = {shape=circle,draw,inner sep=0.4pt}}
\tikzset{edge/.style = {->,> = latex'}}
 
\node[vertex] (1) at  (0,1) {$X_1$};
\node[vertex,accepting] (2) at  (0,-0.3) {$X_2$};
\node[vertex] (4) at  (1.5,1) {$X_4$};
\node[vertex,accepting] (5) at  (1.5,-0.3) {$X_5$};
\node[vertex] (3) at  (3,1) {$X_3$};

\draw[edge] (1) to (2);
\draw[edge] (1) to (2);
\draw[edge] (4) to (5);
\draw[edge] (1) to (5);
\end{tikzpicture}
}
\caption{}  \label{fig_SB2}
\end{subfigure}
\caption{A dataset in (a) subjected to selection bias according to $X_2$ in the DAG in (b). In (a), we can only view the samples in blue in practice. The dataset in (c) is the same dataset in (a) but with some missing values. The variable in \textit{pink} is subject to selection bias according to $X_2$ and $X_5$ as in (d) due to the missing values, while the variable in \textit{green} is subject only to $X_2$ as in (b).}
\end{figure*}

\section{Missingness as Selection Bias on Selection Bias} \label{sec_missing}
We can informally interpret missing values as a type of ``selection bias on selection bias.'' Here, the first layer of selection bias due to $\bm{S}$ refers to the aforementioned measurement of \textit{all variables} in $\bm{O}$ in a preferential selection of the samples. Missing values in turn represent the second layer of selection bias because missing values arise due to the measurement of only a \textit{subset of the variables} in $\bm{O}$ in a preferential subset of the available samples already subject to the selection bias of $\bm{S}$. 

The missingness may more formally arise for many reasons as modeled by the factors $\bm{S}_{O_1} \supseteq \bm{S}$ for $O_1$, $\bm{S}_{O_2} \supseteq \bm{S}$ for $O_2$, and so on for all $p$ variables in $\bm{O}$. We therefore encode the binary missingness status (measured or missing) of any observable $O_i \in \bm{O}$ using the binary \textit{missingness indicators} $\{\bm{S}_{O_i} \setminus \bm{S}\} \subseteq \bm{L}$. Here, we measure the value of $\bm{O}_i$ if and only if $\bm{S}_{O_i}=1$ because we must select a sample when $\bm{S}=1$ and then measure the value of $\bm{O}_i$ when $\{\bm{S}_{O_i} \setminus \bm{S}\}=1$. The preferential selection of samples due to some $\bm{S}_{O_i}$ thus amounts to conditioning on $\bm{S}_{O_i}=1$ similar to the original selection bias case; in other words, we no longer have access to i.i.d. samples from the marginal distribution $\mathbb{P}_{O_i}$ or even $\mathbb{P}_{O_i|\bm{S}=1}$ but rather i.i.d samples from $\mathbb{P}_{O_i|\bm{S}_{O_i}=1}$. We can also consider arbitrary joint distributions $\mathbb{P}_{\bm{V}}$, where $\bm{V} \subseteq \bm{O}$. We have access to $\mathbb{P}_{\bm{V}}$ without selection bias, $\mathbb{P}_{\bm{V}|\bm{S}=1}$ with selection bias, and $\mathbb{P}_{\bm{V}|\bm{S}_{\bm{V}}=1}$ with selection bias and missing values, where $\bm{S}_{\bm{V}} = \cup_{V \in \bm{V}} \bm{S}_V$.

Consider for example the same samples in Figure \ref{table_SB1} but with missing values according to the binary variable $X_5$ in Figure \ref{table_SB2}. Now, the variable $X_3$ highlighted in \textit{pink} in Figure \ref{table_SB2} is subject to the selection variables $\bm{S}_{X_3}=\{X_2, X_5\}=1$ due to the unmeasured or missing values. On the other hand, the variable $X_1$ in \textit{green} is only subject to the original $\bm{S}_{X_1}=\{X_2\}=\bm{S}=1$ because $X_1$ contains no missing values. We can therefore represent these two situations graphically as in Figure \ref{fig_SB} for $\bm{S}_{X_1}$ and Figure \ref{fig_SB2} for $\bm{S}_{X_3}$. Notice that Figure \ref{fig_SB2} has an extra double sided vertex $X_5$ representing the extra conditioning.

Returning to our medical example, $X_4$ may correspond to the number of miles from the hospital to a patient's house. Individuals with low income may have a hard time commuting to the hospital, if they live far away. The physician therefore may not be able to measure SBP $X_3$ in low income patients who live far from the hospital. We thus no longer even have access to i.i.d. samples from $\mathbb{P}_{X_1 X_3|\bm{S}=1}$ but instead have access to i.i.d. samples from a set of conditional distributions $\{\mathbb{P}_{X_1|\bm{S}_{X_1}=1}, \mathbb{P}_{X_3|\bm{S}_{X_3}=1}, \mathbb{P}_{X_1X_3|\bm{S}_{X_1X_3}=1} \}$.

More generally, we do not have access to i.i.d. samples from $\mathbb{P}_{\bm{O}|\bm{S}=1}$ when missing values exist. Instead, we have access to i.i.d. samples from a set of conditional distributions $\{\mathbb{P}_{\bm{V}|\bm{S}_{\bm{V}}=1}, \forall \bm{V} \subseteq \bm{O} \}$. In this sense, we must deal with heterogeneous selection bias induced by $\bm{S}_{O_1}, \dots,$ $\bm{S}_{O_p}$ as opposed to homogeneous selection bias induced by just $\bm{S}$.

\section{An Assumption on the Missingness Mechanisms} \label{sec_ass}
Let $\bm{S}^u$ denote the set of $q \leq p$ unique members of $\{\bm{S}_{O_1}, \dots,$ $\bm{S}_{O_p}\}$. Note that we have so far imposed no restrictions on the causal relations involving $\bm{S}$ or any member of $\bm{S}^u$. From here on, we will continue to impose no restrictions on the causal relations involving $\bm{S}$, but we will impose restrictions on the causal relations involving the elements in the set $\bm{M} = \{\bm{S}_1^u \setminus \bm{S}, \dots, \bm{S}_q^u \setminus \bm{S} \}$. 

Recall that each $M_i \in \bm{M}$ corresponds to a set of missingness indicators, but we can colloquially call each $M_i \in \bm{M}$ a ``missingness mechanism'' because we obtain missing values for some subset of variables $\bm{V} \subseteq \bm{O}$ when (at least) one member of $M_i$ takes on a value of 0. Here, a missingness mechanism often corresponds to a practical issue. For example, we may have three variables in $M_i$ corresponding to three instruments required to perform a measurement. We have $M_i = 1$ when three instruments can perform the measurement but one variable in $M_i$ equals zero when one of the three instruments fails. We may similarly have $M_3 = 1$ when a subject can commute to the hospital and $M_3 = 0$ when the subject cannot commute to the hospital as in the running medical example; thus $M_3 = X_5$ in this case.

Now let $\mathcal{I}_i = \{\{1, \dots, q\} \setminus i \}$ and consider the following assumption:
\begin{assumption} \label{assump1}
Each $M_i \in \bm{M}$ does not contain an ancestor of any variable in $\bm{O} \cup \bm{S}$ or $\cup_{j \in \mathcal{I}_i} M_j $.
\end{assumption}

The above assumption appears technical at first\\ glance, but we can justify it using an inductive argument that reads as follows.\footnote{Recall that justifying MAR or MCAR in real datasets also requires inductive arguments.} First note that we have no missing values if and only if all variables in the sets in $\bm{M}$ take on a value of one. Suppose then that we have a missing value but then manually set all variables in the sets in $\bm{M}$ to one in order to observe the value. Then we do not expect the mere act of observing a value, or equivalently intervening on the sets in $\bm{M}$, to induce changes in (or causally affect) (1) the values of the observable variables $\bm{O}$ in a dataset or (2) the set of available samples determined by $\bm{S}$. In other words, none of the variables in any set in $\bm{M}$ should be an ancestor of any of the variables in $\bm{O}\cup \bm{S}$.

Assumption \ref{assump1} however imposes the extra condition that no missingness mechanism $M_i \in \bm{M}$ contains an ancestor of any variable in $\cup_{j \in \mathcal{I}_i} M_j$. We find the extra assumption reasonable, if an attempt is made to measure each observable variable in $\bm{O}$ for each sample \textit{regardless} of the missingness status of any other variable in $\bm{O}$. This means that the missingness statuses cannot causally affect each other. Now recall that the missingness status of any variable is determined by the variable's missing mechanism. Hence, we can equivalently state that the missingness mechanisms cannot causally affect each other (i.e., each $M_i \in \bm{M}$ does not contain an ancestor of any variables in $\cup_{j \in \mathcal{I}_i} M_j$). For example, we assume that a failure of any one of three instruments does not cause an investigator to potentially forgo the measurement of other variables which do not require the instruments but say rather cost a lot of money. Instead, the investigator attempts to measure the other variables regardless of whether or not an instrument fails. The instrument failures are thus not causes of the inability to pay or any other missingness mechanism. We conclude inductively that Assumption \ref{assump1} is justified with ``comprehensively measured observational data,'' where an attempt is made to measure each observable variable in $\bm{O}$ for each sample regardless of the missingness status of any other variable in $\bm{O}$.

\section{Graph Theory} \label{sec_theory}
We can justify test-wise deletion, if we can utilize Assumption \ref{assump1}. We will consider the set of selection variables $\bm{S}_l = \cup_{i=1}^p \bm{S}_{O_i}$. Notice that we obtain $\bm{S}_l=1$, when we perform list-wise deletion on the dataset. Also let $\bm{S}_{O_iO_j\bm{W}}$ refer to the selection set induced by the complete samples among the variables $O_i, O_j$ and $\bm{W} \subseteq \bm{O} \setminus \{O_i, O_j\}$ alone by setting $\bm{V}=\{O_i, O_j, \bm{W} \}$ in $\bm{S}_{\bm{V}}$; in other words, $\bm{S}_{O_iO_j\bm{W}}$ corresponds to the selection variables obtained after performing \textit{test-wise deletion}, or list-wise deletion only among the variables $\{O_i, O_j, \bm{W} \}$.

The following important lemma now forms the basis of our arguments:
\begin{lemma} \label{lem_CD}
Consider Assumption \ref{assump1}. If $O_i \not \ci_d O_j | (\bm{W},$ $\bm{S}_{O_iO_j\bm{W}})$ with $\bm{W} \subseteq \bm{O} \setminus \{O_i, O_j\}$, then $O_i \not \ci_d O_j | (\bm{W},$ $\bm{S}_l)$.
\end{lemma}
\begin{proof}
If $O_i \not \ci_d O_j | (\bm{W}, \bm{S}_{O_iO_j\bm{W}})$, then $\{\bm{W}, \bm{S}_{O_iO_j\bm{W}} \}$ must contain the descendants of all colliders and no non-colliders on a path $\pi$ between $O_i$ and $O_j$. In other words, $\pi$ is active given $\{\bm{W}, \bm{S}_{O_iO_j\bm{W}}\}$. Note that the conclusion follows trivially if $\bm{S}_l = \bm{S}_{O_iO_j\bm{W}}$. Suppose then that we have $\bm{S}_l \supset \bm{S}_{O_iO_j\bm{W}}$. Let $\bm{T}_{O_iO_j\bm{W}} = \bm{S}_l \setminus \bm{S}_{O_iO_j\bm{W}}$. It suffices to show that $\bm{T}_{O_iO_j\bm{W}}$ cannot contain a non-collider on $\pi$. Suppose for a contradiction that there exists a variable $Z \in \bm{T}_{O_iO_j\bm{W}}$ which is a non-collider on $\pi$. Then $Z$ must be an ancestor of $O_i$, $O_j$ or $\bm{S}_{O_iO_j\bm{W}}$. Note that we have $Z \in \{ \bm{S}_l \setminus \bm{S} \}$, so $Z$ must be a member of at least one element in $\bm{M}$. As a result, $Z$ cannot be an ancestor of any variable in $\bm{O}$ by Assumption \ref{assump1}; thus $Z$ cannot be an ancestor of $O_i$ or $O_j$. So $Z$ can only be an ancestor of $\bm{S}_{O_iO_j\bm{W}}$. The variable $Z$ however cannot be an ancestor of $\bm{S}$ also by Assumption \ref{assump1}, so $Z$ can only be an ancestor of $\bm{S}_{O_iO_j\bm{W}}\setminus \bm{S}$. Next, notice that $Z$ is not a member of $\bm{S}_{O_i}\setminus \bm{S}$, $\bm{S}_{O_j}\setminus \bm{S}$ or $\bm{S}_{W}\setminus \bm{S}$ for any $W \in \bm{W}$ because we have $Z \in \bm{T}_{O_iO_j\bm{W}}$. Thus no element in $\bm{M}$ containing $Z$ can also contain an ancestor of $\bm{S}_{O_i}\setminus \bm{S}$, $\bm{S}_{O_j}\setminus \bm{S}$ or $\bm{S}_{W}\setminus \bm{S}$ for any $W \in \bm{W}$ by Assumption \ref{assump1}; hence no element in $\bm{M}$ containing $Z$ can also contain an ancestor of $\bm{S}_{O_iO_j\bm{W}}\setminus \bm{S}$, so $Z$ cannot be an ancestor of $\bm{S}_{O_iO_j\bm{W}}\setminus \bm{S}$. We conclude by contradiction that $\bm{T}_{O_iO_j\bm{W}}$ cannot contain a non-collider on $\pi$. Hence, we have $O_i \not \ci_d O_j | (\bm{W},$ $\bm{S}_l)$ via the active path $\pi$. \qed
\end{proof}
The above lemma leads to important conclusions regarding the design of a CCD algorithm. We begin to justify v-structure discovery and an orientation rule (R1) using test-wise deletion with another lemma:
\begin{lemma} \label{lem_comb}
Consider Assumption \ref{assump1}. If we have $O_i \ci_d O_j | (\bm{W}, \bm{S}_{O_iO_j\bm{W}})$ with $\bm{W} \subseteq \bm{O} \setminus \{O_i, O_j\}$ minimal and $O_i \ci_d O_j |$ $(\bm{W}, \bm{S}_l)$, then $O_i \ci_d O_j | (\bm{W}, \bm{S}_l)$ with $\bm{W}$ minimal.
\end{lemma}
\begin{proof}
If we have $O_i \ci_d O_j | (\bm{W}, \bm{S}_{O_iO_j\bm{W}})$ with $\bm{W} \subseteq \bm{O} \setminus \{O_i, O_j\}$ minimal, then $O_i \not \ci_d O_j | (\bm{A}, \bm{S}_{O_iO_j\bm{W}})$, where $\bm{A}$ denotes an arbitrary strict subset of $\bm{W}$. Now $O_i$ $\not \ci_d O_j | (\bm{A}, \bm{S}_{O_iO_j\bm{W}})$ implies $O_i \not \ci_d O_j | (\bm{A}, \bm{S}_l)$ by Lemma \ref{lem_CD}. The conclusion follows because we chose $\bm{A}$ arbitrarily and assumed $O_i \ci_d O_j |(\bm{W}, \bm{S}_l)$. \qed
\end{proof}
Note that FCI and RFCI require many calls to a CI oracle precisely because they search for minimal separating sets. We can therefore take advantage of Lemma \ref{lem_comb} by searching for minimal separating sets with test-wise deletion and then only confirming the separating sets with list-wise deletion (rather than directly searching for minimal separating sets with list-wise deletion). This greatly reduces the number of CI tests performed with list-wise deletion.

We can now directly justify some desired conclusions. Let us examine the most difficult arguments in detail. We have the following conclusion for v-structure discovery and R1:
\begin{proposition} \label{thm_v_R1}
Consider Assumption \ref{assump1}. Suppose $O_i$\\ $\ci_d O_j | (\bm{W}, \bm{S}_{O_iO_j\bm{W}})$ with $\bm{W} \subseteq \bm{O} \setminus \{O_i, O_j\}$ minimal and $O_i \ci_d O_j | (\bm{W}, \bm{S}_l)$. Further assume $O_i   \not \ci_d O_k | (\bm{W}, \bm{S}_{O_iO_j\bm{W}})$ and $O_j   \not \ci_d O_k | (\bm{W}, $ $\bm{S}_{O_iO_j\bm{W}})$. We have $O_k \in \bm{W}$ if and only if $O_k \in \bm{An}(\{O_i, O_j\} \cup \bm{S}_l)$.
\end{proposition}
\begin{proof}
If $O_i \ci_d O_j | (\bm{W}, \bm{S}_{O_iO_j\bm{W}})$ with $\bm{W}$ minimal and $O_i \ci_d O_j | (\bm{W}, \bm{S}_l)$, then $O_i \ci_d O_j | (\bm{W}, \bm{S}_l)$ with $\bm{W}$ minimal by Lemma \ref{lem_comb}. By Lemma \ref{lem_CD}, we know that $O_i   \not \ci_d O_k | (\bm{W}, \bm{S}_{O_iO_j\bm{W}})$ and $O_j   \not \ci_d O_k | (\bm{W}, \bm{S}_{O_iO_j\bm{W}})$ imply $O_i   \not \ci_d O_k | (\bm{W}, \bm{S}_l)$ and $O_j   \not \ci_d O_k | (\bm{W}, \bm{S}_l)$, respectively. We may now invoke Lemma 3.1 of \cite{Colombo12} with $\bm{S}_l$.\qed
\end{proof}
We can also justify the discriminating path rule (R4) with a similar argument:
\begin{proposition} \label{thm_R4}
Consider Assumption \ref{assump1}. Let $\pi_{ik} =$\\ $\{ O_i, \dots, O_l, O_j, O_k \}$ be a sequence of at least four vertices which satisfy the following:
\begin{enumerate}
\item $O_i \ci_d O_k | (\bm{W}, \bm{S}_l)$ with  $\bm{W} \subseteq \bm{O} \setminus \{O_i, O_j\}$,
\item Any two successive vertices $O_h$ and $O_{h+1}$ on $\pi_{ik}$ are d-connected given: 
\begin{equation} \nonumber
(\bm{Y} \setminus \{O_h, O_{h+1} \} ) \cup\bm{S}_{O_{h}O_{h+1}(\bm{Y} \setminus \{O_h, O_{h+1} \} )}
\end{equation} 
for all $\bm{Y} \subseteq \bm{W}$,
\item All vertices $O_h$ between $O_i$ and $O_j$ (not including $O_i$ and $O_j$) satisfy $O_h \in \bm{An}(O_k)$ and $O_h \not \in$\\ $\bm{An}(\{ O_{h-1}, O_{h+1} \}\cup \bm{S}_l)$, where $O_{h-1}$ and $O_{h+1}$ denote the vertices adjacent to $O_h$ on $\pi_{ik}$.
\end{enumerate}
Then, if $O_j \in \bm{W}$, then $O_j \in \bm{An}(O_k \cup \bm{S}_l)$ and $O_k \in \bm{An}(O_j \cup \bm{S}_l)$. On the other hand, if $O_j \not \in \bm{W}$, then $O_j \not \in \bm{An}(\{O_l, O_k\} \cup \bm{S}_l)$ and $O_k \not \in \bm{An}(O_j \cup \bm{S}_l)$.
\end{proposition}
\begin{proof}
Apply Lemma \ref{lem_CD} to conclude that the required d-connections with $\bm{S}_{O_{h}O_{h+1}(\bm{Y} \setminus \{O_h, O_{h+1} \} )}$ also hold with $\bm{S}_l$. Subsequently invoke Lemma 3.2 of \cite{Colombo12} with $\bm{S}_l$. \qed
\end{proof}
\noindent We can prove the soundness of the remaining orientation rules in \cite{Zhang08} using a similar strategy. Many desired conclusions therefore easily follow from Lemmas \ref{lem_CD} and \ref{lem_comb}.

\section{Algorithms with Test-Wise Deletion} \label{sec_alg}

We introduced some direct proofs in Propositions \ref{thm_v_R1} and \ref{thm_R4} of the previous section which capitalize on Lemmas \ref{lem_CD} and \ref{lem_comb}. We can however actually prove the \textit{full} soundness and completeness of FCI in one sweep by designing a CI oracle wrapper.

We introduce the CI oracle wrapper in Algorithm \ref{alg_wrapper}. Wlog, we assume that the CI oracle outputs 0 when conditional dependence holds and 1 otherwise. The wrapper works by first querying the CI oracle with $\bm{S}_{O_i O_j \bm{W}}$ in line \ref{alg:first_query}. If the CI oracle outputs 1, then the wrapper also checks whether the CI oracle outputs 1 with $\bm{S}_l$ in line \ref{alg:second_query}. If so, the wrapper outputs 1 and otherwise outputs 0 due to line \ref{alg:combine}. Hence, the wrapper claims conditional independence only when both the CI oracle with $\bm{S}_{O_i O_j \bm{W}}$ and the CI oracle with $\bm{S}_l$ output 1 thus implementing Lemma \ref{lem_comb}. On the other hand, if the CI oracle with $\bm{S}_{O_i O_j \bm{W}}$ outputs 0, then the wrapper immediately outputs 0 thus implementing Lemma \ref{lem_CD}. Notice that the finite sample version of Algorithm \ref{alg_wrapper} follows immediately by replacing line \ref{line_alpha} with an $\alpha$ level cutoff.

{\linespread{0.8}\selectfont
\begin{algorithm}[] \label{fci_vstruc}
 \KwData{$O_i$, $O_j$, $\bm{W}$}
 \KwResult{$p$}
 \BlankLine
 
 $p \leftarrow $ Ask CI oracle whether $O_i \ci O_j | (\bm{W}, \bm{S}_{O_i O_j \bm{W}})$ \label{alg:first_query}\\
 \If{$p$ is $1$ \label{line_alpha} \label{alg:check_query}}{
 	 $q \leftarrow $ Ask CI oracle whether $O_i \ci O_j | (\bm{W}, \bm{S}_l)$ \label{alg:second_query} \\
     $p \leftarrow \min(p,q)$ \label{alg:combine}
 }
 
 \BlankLine

 \caption{CI oracle wrapper} \label{alg_wrapper}
\end{algorithm}
}

We now make the following claim:
\begin{theorem} \label{thm1}
Consider Assumption \ref{assump1}. Further assume d-separation faithfulness. Then FCI using Algorithm \ref{alg_wrapper} outputs the same graph as FCI using a CI oracle with $\bm{S}_l$. The same result holds for RFCI.
\end{theorem}
\begin{proof}
Recall that d-separation and conditional independence are equivalent under d-separation faithfulness, so we can talk about d-separation and conditional independence interchangeably. It then suffices to show that $O_i \not \ci O_j | (\bm{W}, \bm{S}_l)$ if and only if Algorithm \ref{alg_wrapper} outputs zero. For the backward direction, if Algorithm \ref{alg_wrapper} outputs zero, then we must have (1) $O_i\not \ci O_j | (\bm{W},$ $\bm{S}_{O_i O_j | \bm{W}})$ or (2) $O_i \not \ci O_j | (\bm{W}, \bm{S}_l)$ (or both). If (2) holds, the conclusion follows immediately. If (1) holds, then the conclusion follows by Lemma \ref{lem_CD}. For the forward direction, assume for a contrapositive that Algorithm \ref{alg_wrapper} outputs one. The conclusion follows because Algorithm \ref{alg_wrapper} outputs one only if we have $O_i \ci O_j | (\bm{W},$ $\bm{S}_{O_i O_j \bm{W}})$ and $O_i \ci O_j | (\bm{W}, \bm{S}_l)$. \qed
\end{proof}
It immediately follows that FCI equipped with Algorithm \ref{alg_wrapper} is sound and complete under d-separation faithfulness and Assumption \ref{assump1}. 

\section{Experiments} \label{sec_exps}
We now describe the experiments used to assess the finite sample size performance of test-wise deletion as compared to existing approaches.
\subsection{Algorithms}
We compared the following algorithms with Fisher's z-test and $\alpha$ set to 0.01:
\begin{enumerate}
\item FCI with test-wise deletion (i.e., equipped with Algorithm \ref{alg_wrapper});
\item FCI with heuristic test-wise deletion, where we only run line \ref{alg:first_query} of Algorithm \ref{alg_wrapper}. We can justify this procedure under values missing completely at random (MCAR) as explained in detail in Appendix \ref{appendix_Htest};
\item FCI with list-wise deletion \cite{Spirtes01};
\item FCI with five different imputation methods including hot deck \cite{Cranmer13}, k-nearest neighbor with (k=5; k-NN) \cite{Kowarik16}, Bayesian linear regression (BLR) \cite{VanBuuren11,Brand99,Schafer97}, predictive mean matching (PMM) \cite{Little88,Rubin87,VanBuuren05,VanBuuren11} and random forests (ntree=10; RF) \cite{Doove14,Shah14,VanBuuren12}.
\end{enumerate}
We then repeated the comparisons with RFCI in place of FCI. We therefore compared a total of 16 methods.

Note that FCI with heuristic test-wise deletion is not justified in the general MNAR case, but we find that it performs well with finite sample CI tests and hence report its results mainly in the Appendix.

\subsection{Synthetic Data}
\subsubsection{Data Generation}
We used the following procedure in \cite{Colombo12} to generate 400 different Gaussian DAGs with an expected neighborhood size of $\mathbb{E}(N)=2$ and $p=20$ vertices. First, we generated a random adjacency matrix $\mathcal{A}$ with independent realizations of $\text{Bernoulli}(\mathbb{E}(N)/(p - 1))$ random variables in the lower triangle of the matrix and zeroes in the remaining entries. Next, we replaced the ones in $\mathcal{A}$ by independent realizations of a $\text{Uniform}([-1,-0.1]\cup[0.1, 1])$ random variable. We can interpret a nonzero entry $\mathcal{A}_{ij}$ as an edge from $X_i$ to $X_j$ with
coefficient $\mathcal{A}_{ij}$ in the following linear model:
\begin{equation}
\begin{aligned}
&X_1 = \varepsilon_1,\\
&X_i = \sum_{r=1}^{p-1} \mathcal{A}_{ir}X_r + \varepsilon_i,
\end{aligned}
\end{equation}
for $i = 2, \dots , p$ where $\varepsilon_1, . . ., \varepsilon_p$ are mutually independent $\mathcal{N}(0, 1)$ random variables. We finally introduced non-zero means $\mu$ by adding $p$ independent realizations of a $\mathcal{N}(0,4)$ random variable to $\bm{X}$. The variables $X_1,\dots,$ $X_p$ then have a multivariate Gaussian distribution with mean vector $\mu$ and covariance matrix $\Sigma = (\mathbb{I} - \mathcal{A})^{-1}(\mathbb{I} - \mathcal{A})^{-T}$, where $\mathbb{I}$ is the $p \times p$ identity matrix.

We generated MNAR datasets using the following procedure. We first randomly selected a set of 0-4 latent common causes $\bm{L}$ without replacement. We then selected a set of 1-2 additional latent variables $\widetilde{\bm{L}}$ without replacement from the set $\bm{X} \setminus \bm{L}$. Next, we randomly selected a subset of 3-6 variables in $\bm{O}=\{\bm{X} \setminus \{ \bm{L}, \widetilde{\bm{L}} \}\}$ without replacement for each $\widetilde{L} \in \widetilde{\bm{L}}$ and then removed the bottom $r$ percentile of samples from those 3-6 variables according to $\widetilde{L}$; we drew $r$ according to independent realizations of a $\text{Uniform}([0.1,0.5])$ random variable. Thus, the missing values depend directly on the unobservables $\widetilde{\bm{L}}$ in this MNAR case. Assumption \ref{assump1} is also satisfied because none of the missingness indicators have children. We finally eliminated all of the instantiations of the latent variables $\bm{L} \cup \widetilde{\bm{L}}$ from the dataset. 

For the MAR case, we again randomly selected a set of 0-4 latent common causes $\bm{L}$ (at least two children) without replacement. We then selected a set of 1-2 observable variables $\widetilde{\bm{O}} \subseteq \bm{O}=\{\bm{X} \setminus \bm{L} \}$ without replacement and then randomly selected a subset of 3-6 variables in $\bm{O} \setminus \widetilde{\bm{O}}$ without replacement for each variable in $\widetilde{O} \in \widetilde{\bm{O}}$. We next removed the bottom $r$ percentile of samples from those 3-6 variables according to $\widetilde{O}$; we again drew $r$ according to independent realizations of a $\text{Uniform}([0.1,0.5])$ random variable. Thus, the missing values depend directly on the observables $\widetilde{\bm{O}}$ with no missing values in this MAR case. We finally again eliminated all of the instantiations of the latent variables $\bm{L}$ from the dataset. 

We ultimately created datasets with sample sizes of 100, 250, 500, 1000 and 5000 for each of the 400 DAGs for both the MNAR and MAR cases. We therefore generated a total of $400 \times 5 \times 2 = 4000$ datasets.

\subsubsection{Metrics}
We compared the algorithms using the structural Hamming distance (SHD) from the oracle graphs in the MNAR and MAR cases. 

We set the selection variables to $\bm{S}_l$ for the oracle graphs in the MNAR case, since FCI with test-wise and list-wise deletion recover these graphs in the sample limit. Here, we hope FCI and RFCI with test-wise deletion will outperform FCI and RFCI with list-wise deletion, respectively, by obtaining lower SHD scores on average.

We also set the selection variables to the empty set for the oracle graphs in the MAR case, since a sound imputation method should recover the underlying distribution without selection bias using the observed values. Note however that test-wise deletion cannot eliminate the selection bias induced by the observed values in this case. Clearly then the imputation methods should outperform test-wise deletion under the metric of SHD to the oracle graph without selection bias. However, we still hope that the algorithms with test-wise deletion will perform reasonably well because none of variables in $\widetilde{\bm{O}}$ induce dense MAGs by design when acting as selection variables.

\subsubsection{Results}
\begin{figure*}
\centering
\begin{subfigure}{0.4\textwidth}
  \centering
  \includegraphics[width=0.8\linewidth]{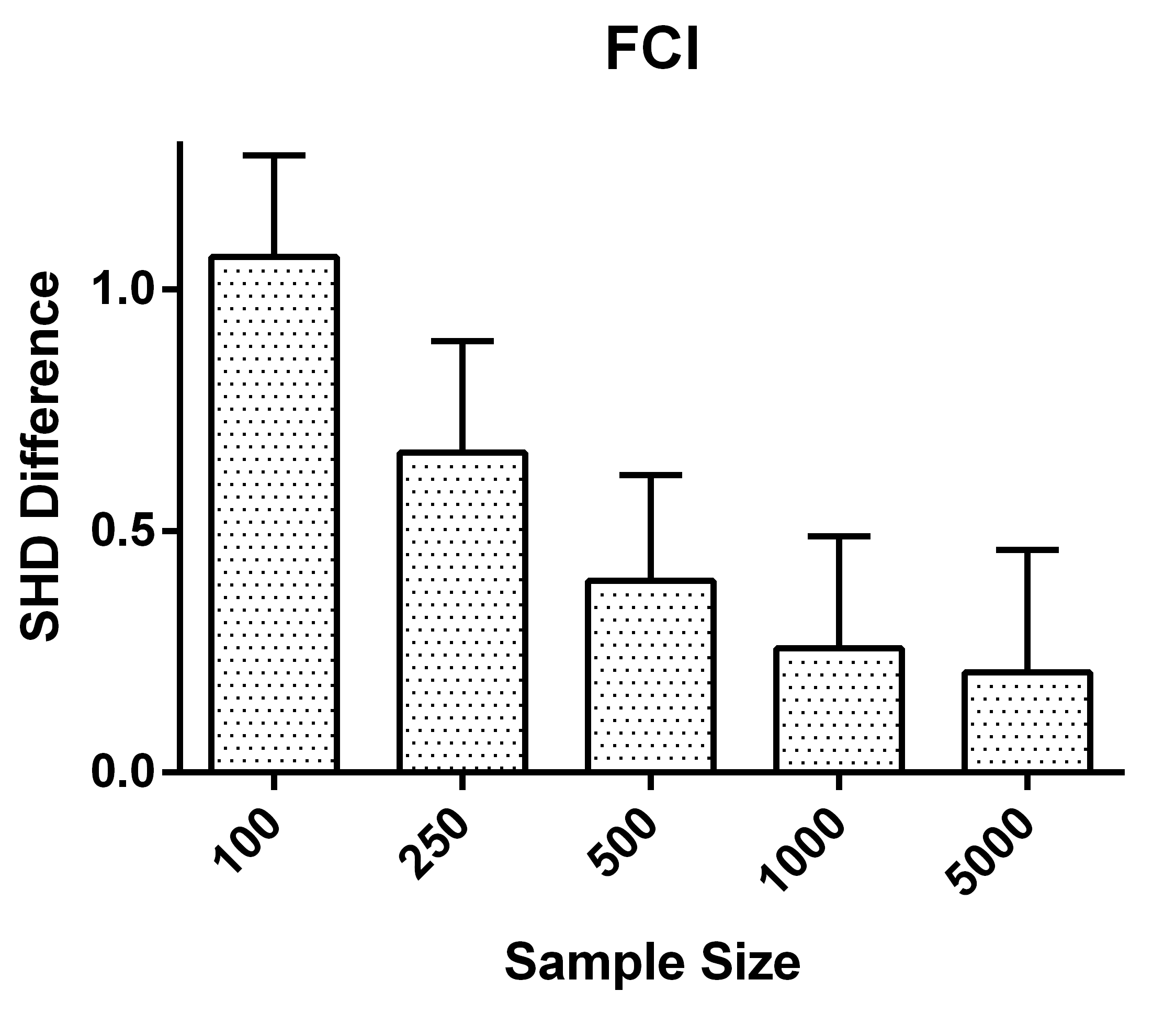}
  \caption{}
  \label{fig_MNAR:FCI_SHD}
\end{subfigure}
\begin{subfigure}{0.4\textwidth}
  \centering
  \includegraphics[width=0.8\linewidth]{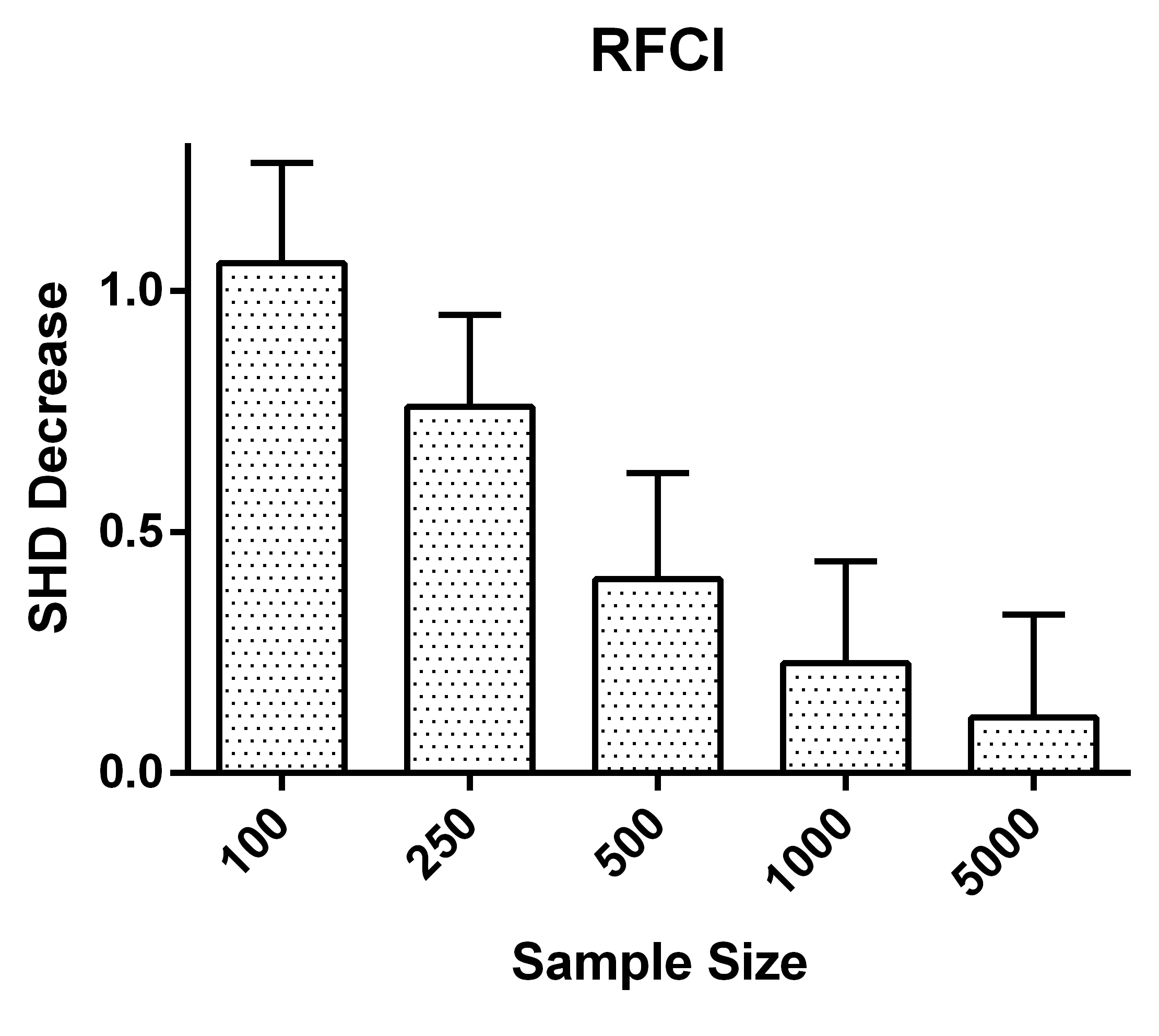}
  \caption{}
  \label{fig_MNAR:RFCI_SHD}
\end{subfigure}

\begin{subfigure}{0.4\textwidth}
  \centering
  \includegraphics[width=0.8\linewidth]{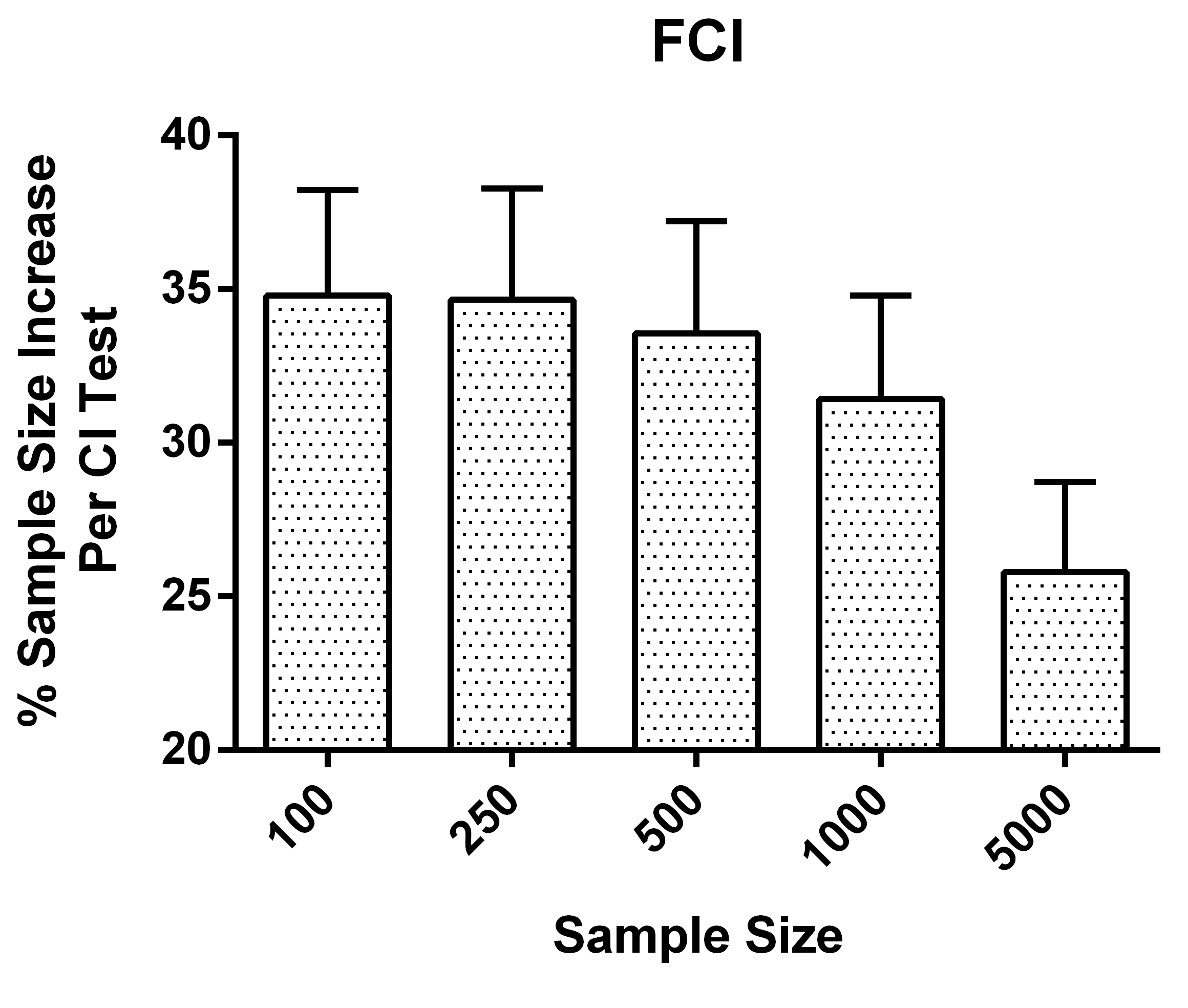}
  \caption{}
  \label{fig_MNAR:FCI_Samples}
\end{subfigure}
\begin{subfigure}{0.4\textwidth}
  \centering
  \includegraphics[width=0.8\linewidth]{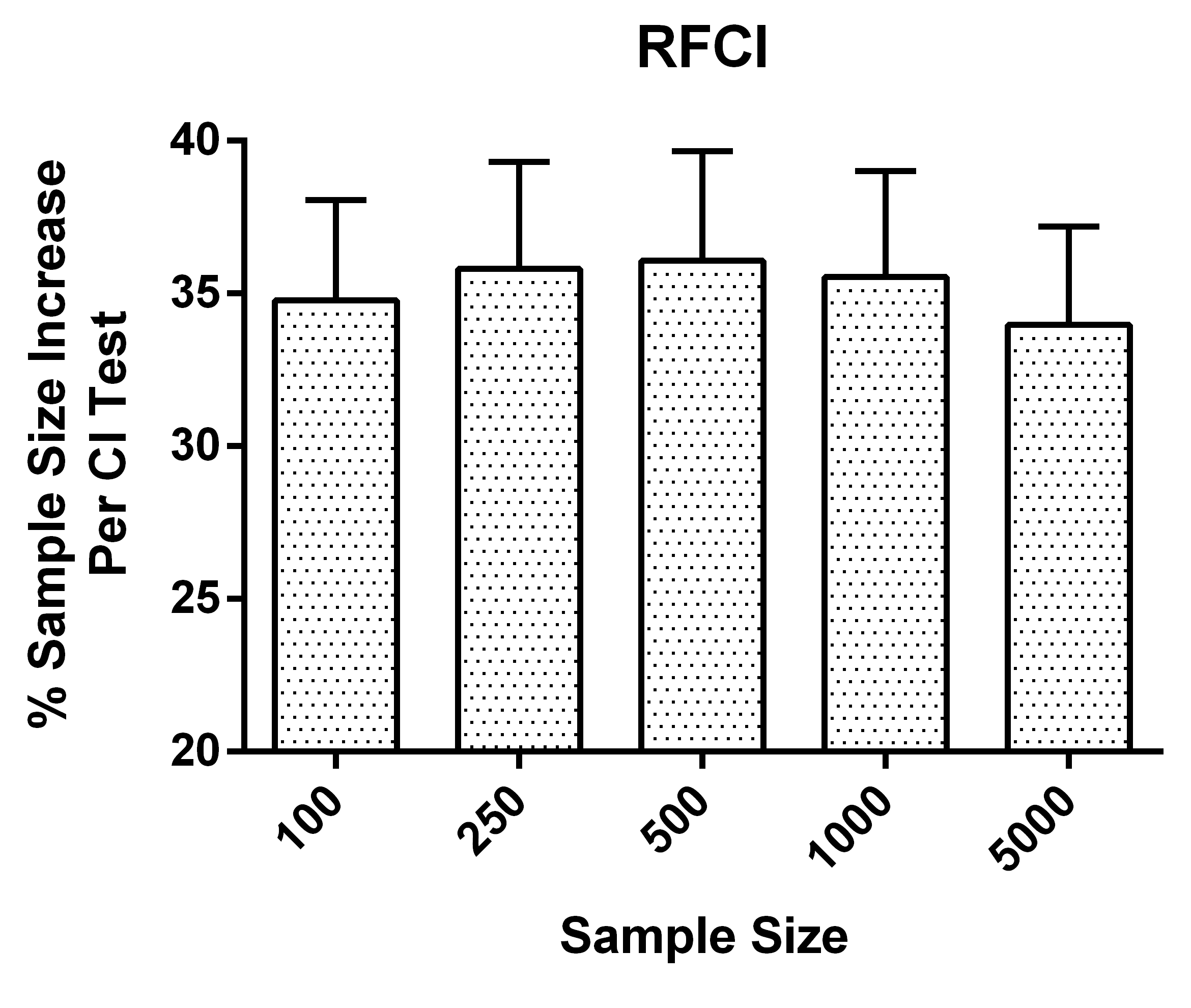}
  \caption{}
  \label{fig_MNAR:RFCI_Samples}
\end{subfigure}
\caption{FCI and RFCI with test-wise deletion vs. the same algorithms with list-wise deletion in the MNAR case. Test-wise deletion results in a decrease in the average SHD for FCI in (a) and RFCI in (b). The performance increase results because of a 25-35\% increase in sample size per CI test on average for FCI in (c) and RFCI in (d).} \label{fig_MNAR}
\end{figure*}

\begin{figure*}
\centering
\begin{subfigure}{0.4\textwidth}
  \centering
  \includegraphics[width=0.8\linewidth]{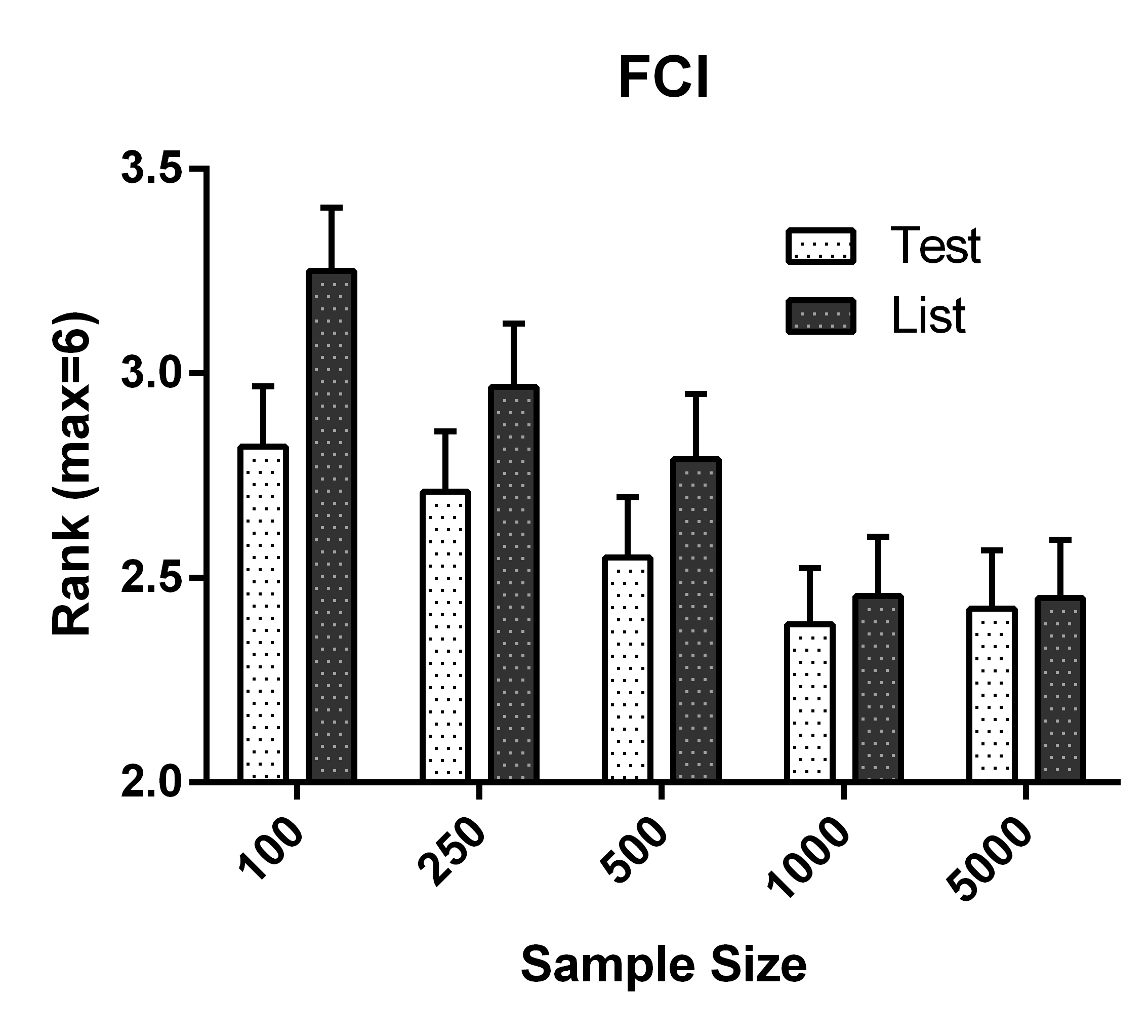}
  \caption{}
  \label{fig_MAR:FCI_rank}
\end{subfigure}
\begin{subfigure}{0.4\textwidth}
  \centering
  \includegraphics[width=0.8\linewidth]{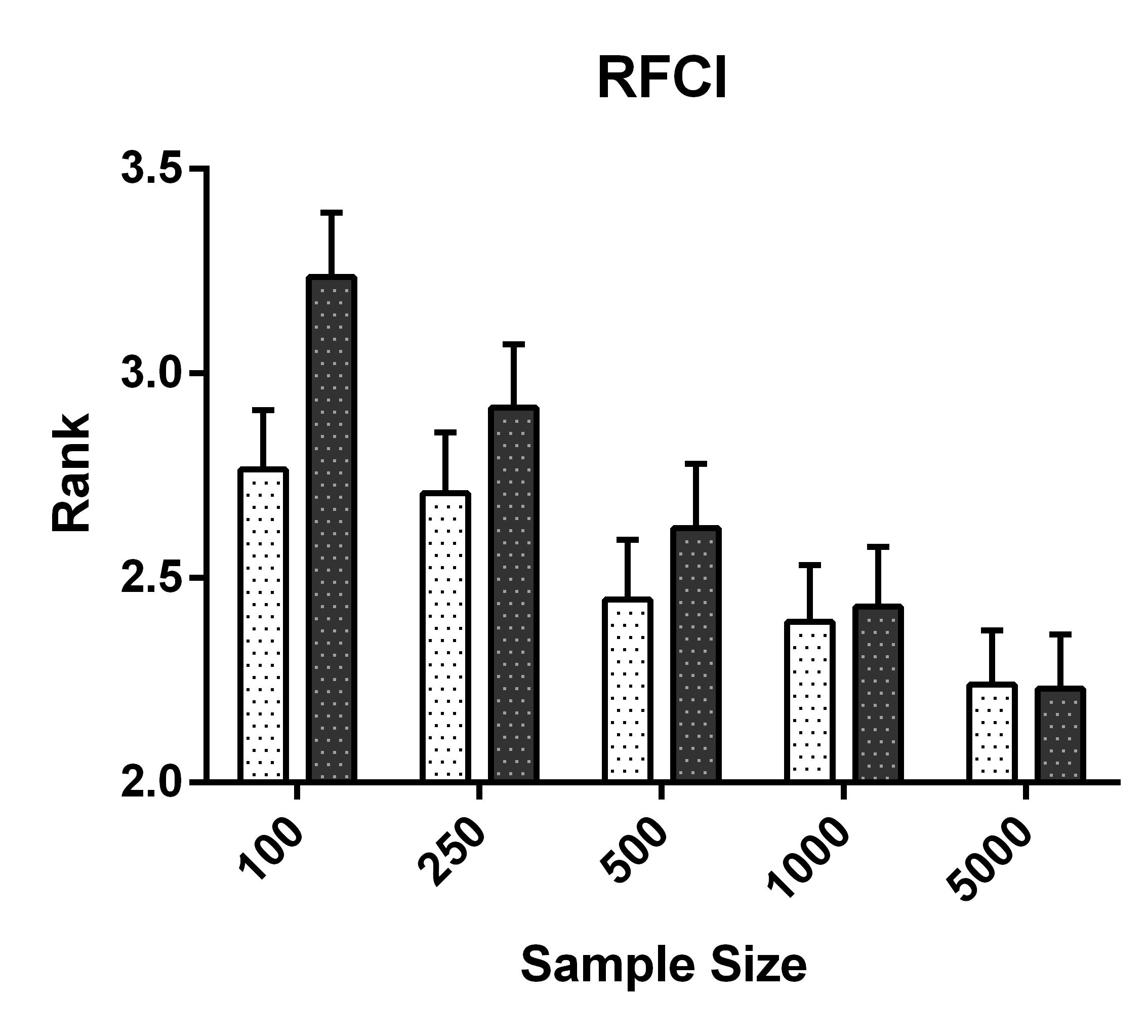}
  \caption{}
  \label{fig_MAR:RFCI_rank}
\end{subfigure}

\begin{subfigure}{0.4\textwidth}
  \centering
  \includegraphics[width=0.8\linewidth]{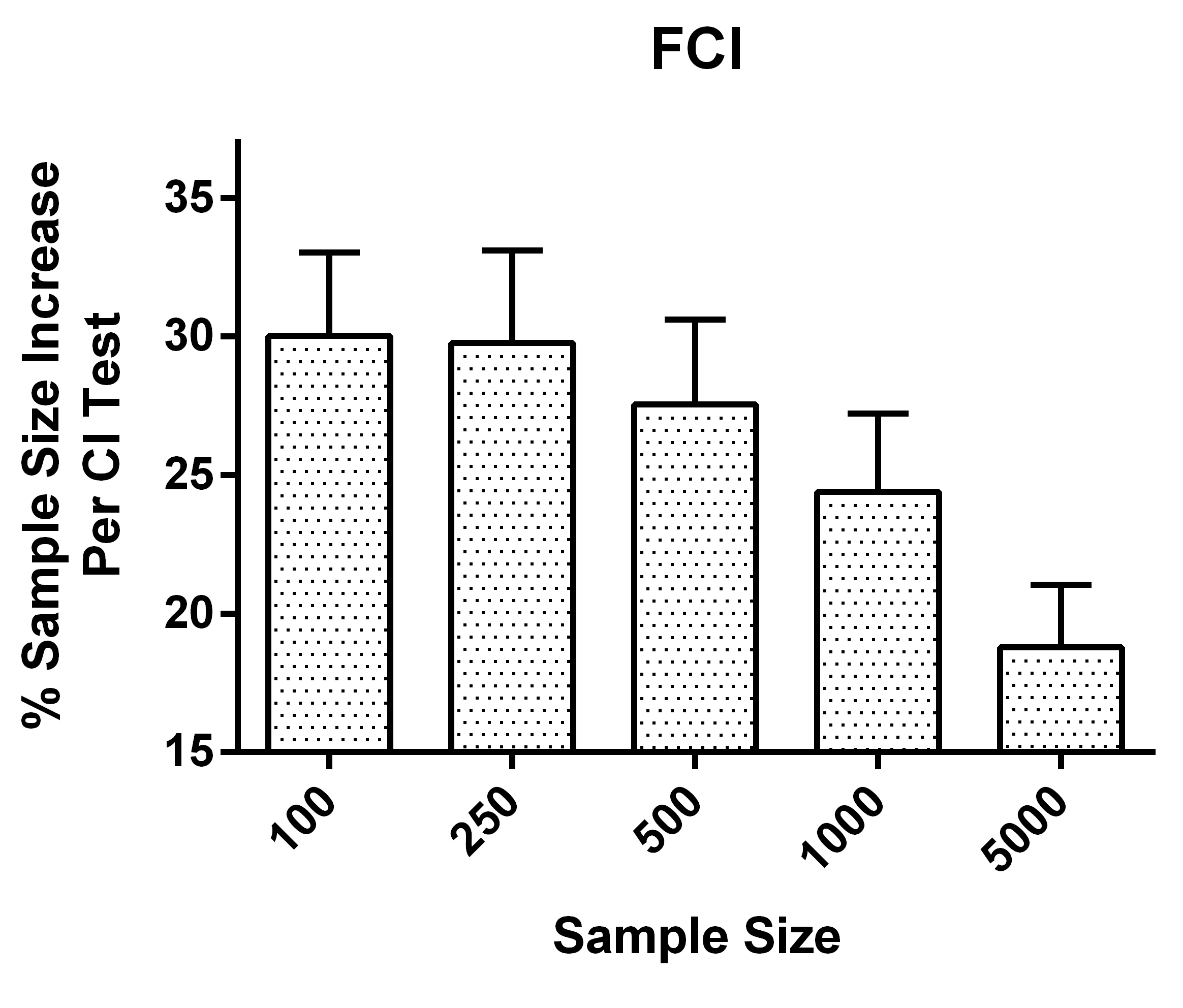}
  \caption{}
  \label{fig_MAR:FCI_Samples_imp}
\end{subfigure}
\begin{subfigure}{0.4\textwidth}
  \centering
  \includegraphics[width=0.8\linewidth]{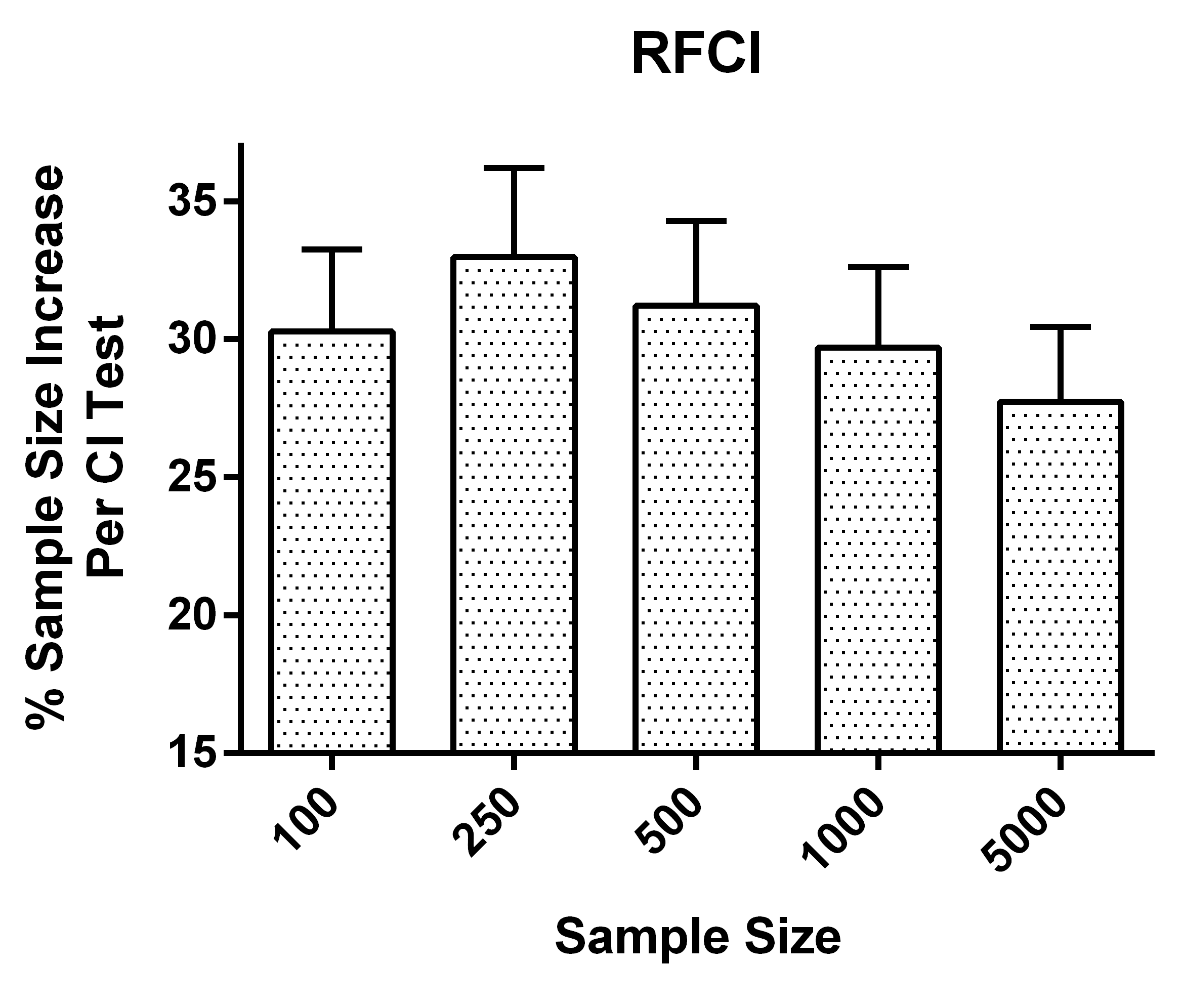}
  \caption{}
  \label{fig_MAR:RFCI_Samples_imp}
\end{subfigure}
\caption{Test-wise deletion vs. list-wise deletion as compared to five imputation methods in the MAR case. A rank of one denotes the best performance and a rank of six denotes the worst. Test-wise deletion has a smaller average rank than list-wise deletion for FCI in (a) and RFCI in (b). The performance increase of test-wise deletion again results because of increased sample efficiency for FCI in (c) and RFCI in (d).} \label{fig_MAR}
\end{figure*}

We have summarized the results for the MNAR case in Figure \ref{fig_MNAR}. We focus on comparing FCI and RFCI with test-wise deletion against the same algorithms with list-wise deletion. FCI and RFCI with any of the five imputation methods expectedly performed much worse in this task, so we relegate the imputation results to Figure \ref{fig_imp} in the Appendix. We have also summarized the excellent results of heuristic test-wise deletion in Figure \ref{fig_heur} in the Appendix. 

Figures \ref{fig_MNAR:FCI_SHD} and \ref{fig_MNAR:RFCI_SHD} suggest that the algorithms with test-wise deletion consistently outperform their list-wise deletion counterparts across all sample sizes. In fact, most test-wise vs. list-wise comparisons were significant using paired t-tests at a Bonferroni corrected threshold of 0.05/5 for both FCI and RFCI (exceptions: FCI at sample size 5000, t=2.018, p=0.044; RFCI 1000, t=2.108, p=0.036; RFCI 5000, t=1.055, p=0.292). We found the largest gains with smaller sample sizes, where the CCD algorithms are prone to error and greatly benefit from the sample size increase (sample size vs. SHD difference correlation; FCI: Pearson's r=-0.0941, t=-4.226, p=2.48E-5; RFCI: r = -0.111, t=-5.000, p=6.24E-7). Figure \ref{fig_MNAR:FCI_Samples} and \ref{fig_MNAR:RFCI_Samples} list the average sample size increase per executed CI test for test-wise deletion compared to list-wise deletion in percentage points. We see that test-wise deletion results in an approximately $25$ to $35\%$ increase in sample size than list-wise deletion regardless of the algorithm. RFCI in particular benefits the most at a steady $35\%$ regardless of the sample size because the algorithm utilizes smaller conditioning set sizes than FCI, so test-wise deletion results in the deletion of even fewer samples per CI test than list-wise deletion for RFCI. We conclude that FCI and RFCI with test-wise deletion consistently outperform their list-wise deletion counterparts because Algorithm \ref{alg_wrapper} allows the algorithms to more efficiently utilize the available samples.

We have also summarized the results for the MAR case in Figure \ref{fig_MAR}. Figures \ref{fig_MAR:FCI_rank} and \ref{fig_MAR:RFCI_rank} list the average ranked results against the five imputation methods. A rank of one denotes the best performance whereas a rank of six denotes the worst. We see that test-wise deletion again outperforms list-wise deletion. The effect is significant for all sample sizes between 100 and 500 at a Bonferonni level of 0.05/5 for FCI (max t = -5.223, max p = 2.84E-7) and for RFCI (max t = -4.147, max p = 4.12E-5). The performance improvements result from the improved sample efficiency of test-wise deletion as compared to list-wise deletion in both FCI (approx. $15$ to $30\%$ increase; Figure \ref{fig_MAR:FCI_Samples_imp}) and RFCI (approx. $30\%$ increase; Figure \ref{fig_MAR:RFCI_Samples_imp}).

Test-wise deletion and list-wise deletion also perform very well overall even in the MAR case as compared to imputation methods. Both methods perform approximately middle of the road (rank approx. $2.5$) and are only consistently outperformed by BLR and PMM, where the linear models are correctly specified. On the other hand, the non-parametric k-NN and random forest imputation methods often fall short of both test-wise and list-wise deletion. We conclude that FCI and RFCI with test-wise deletion are competitive against the same algorithms with imputation even when MAR strictly holds.

\subsection{Real Data}

\begin{figure*}
\centering
\begin{subfigure}{0.4\textwidth}
  \centering
  \includegraphics[width=0.8\linewidth]{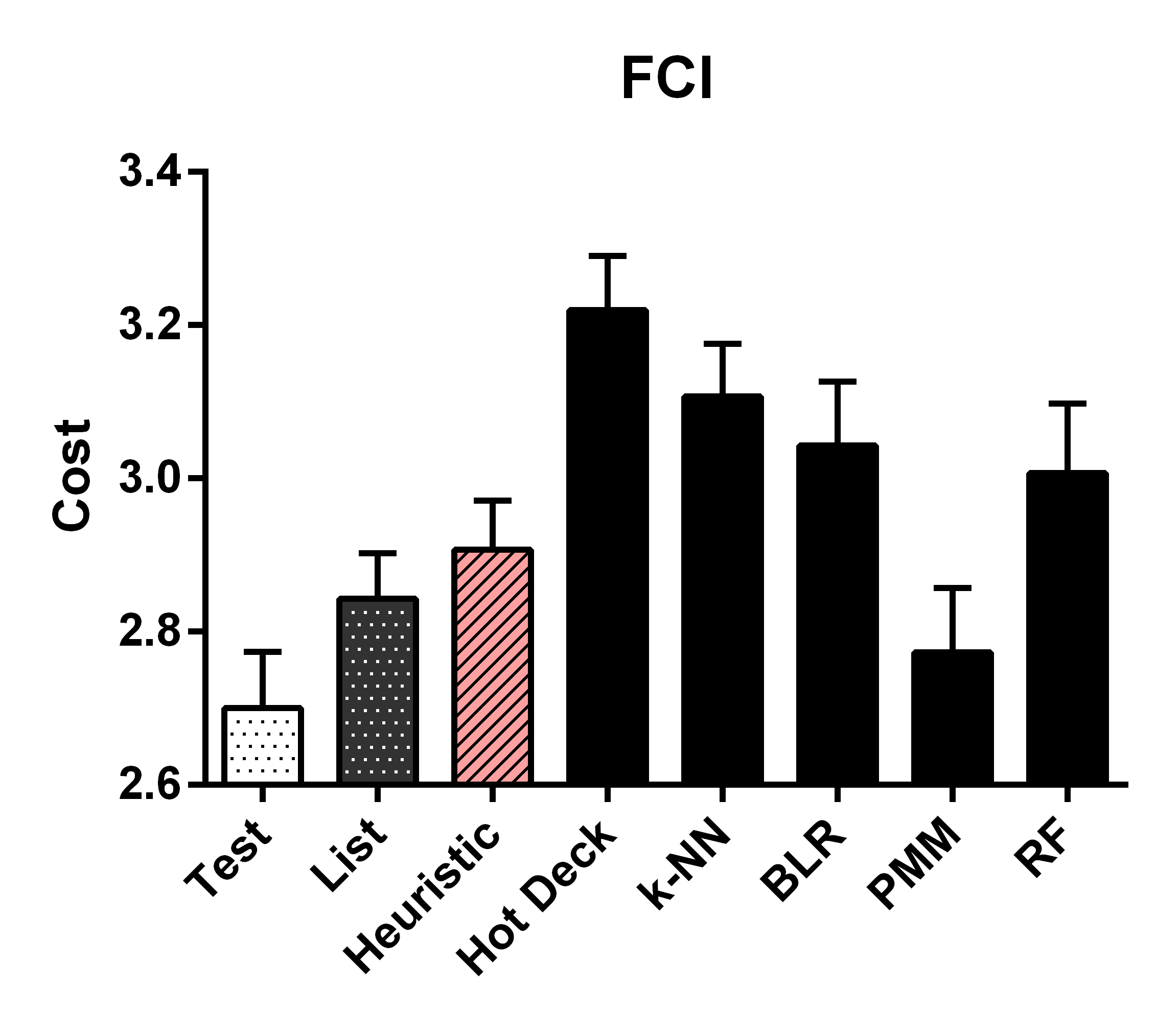}
  \caption{}
  \label{fig_real:Real_FCI}
\end{subfigure}
\begin{subfigure}{0.4\textwidth}
  \centering
  \includegraphics[width=0.8\linewidth]{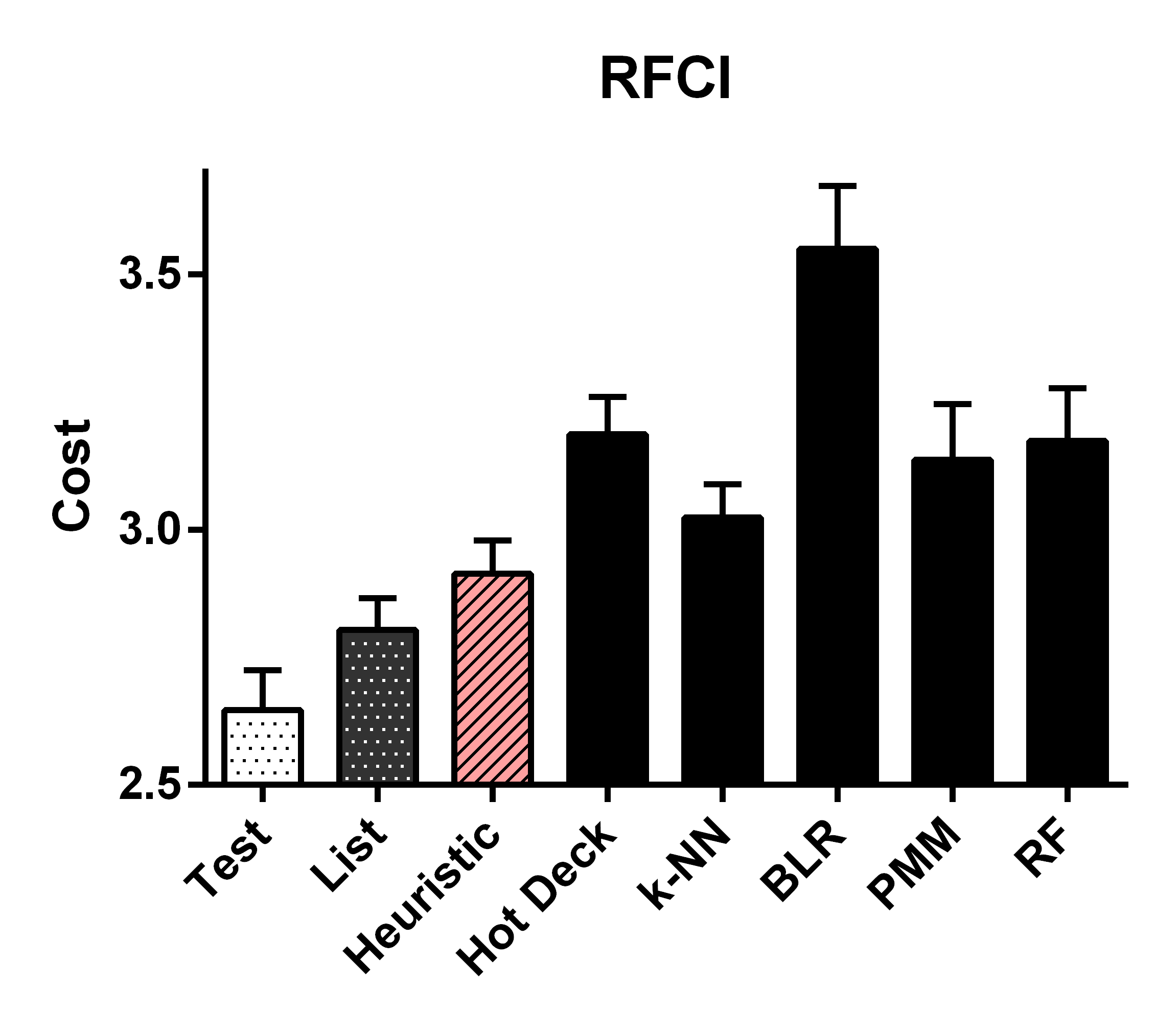}
  \caption{}
  \label{fig_real:Real_RFCI}
\end{subfigure}

\caption{Real data results for all methods in terms of the cost metric $w+v$. Test and list-wise deletion both perform well, but test-wise deletion performs the best when incorporated into both (a) FCI and (b) RFCI.} \label{fig_real}
\end{figure*}

We finally ran the same algorithms using the nonparametric CI test called RCIT \cite{Strobl17} at $\alpha=0.01$ on a publicly available longitudinal dataset from the Cognition and Aging USA (CogUSA) study \cite{McArdle15}, where scientists measured the cognition of men and women above 50 years of age. The dataset contains three waves of data, but we specifically focused on the first two waves in this dataset. The first two waves are only separated by one week, and the investigators collected data for the first wave by telephone. Note that neuropsychological interventions are near impossible within a week after phone-based testing, so no missingness indicator should be an ancestor of $\bm{O} \cup \bm{S}$. Moreover, to the best of our knowledge, the investigators attempted to measure each variable regardless of the missingness statuses of the other variables for each sample. We can therefore justify Assumption \ref{assump1} in this setting.

We used a cleaned version of the dataset containing 1514 samples over 16 variables; we specifically removed deterministic relations and variables related to administrative purposes as opposed to neuropsychological variables. Despite the cleaning, the dataset contains many missing values. List-wise deletion drops the number of samples from 1514 to 1106. However, this is also precisely the setting where we hope to use test-wise deletion in order to increase sample efficiency.

Note that we do not have access to a gold standard solution set in this case. However, we can develop an approximate solution set by utilizing two key facts. First, recall that we cannot have ancestral relations directed backwards in time. Thus, a variable in wave 2 cannot be an ancestor of a variable in wave 1; we can therefore count the number of edges between wave 1 and wave 2 with both a tail and an arrowhead at a vertex in wave 2. Second, the mental status score is a composite score that includes backwards counting as well as some other metrics. Thus, there should exist an edge between backwards counting and mental status, and the edge ideally should have a tail at backwards counting as well as an arrowhead at mental status in both waves.

We used the above solution set to construct the following cost metric; we counted the number of incorrect ancestral relations $w$ as well as counted the number of unoriented or incorrectly oriented endpoints between backwards counting and mental status $v$. A lower cost of $w+v$ therefore indicates better performance. 

We have summarized the results in Figure \ref{fig_real} after generating 300 bootstrapped datasets. Test-wise deletion outperforms 6 of the 7 other methods at a Bonferroni corrected threshold of 0.05/7 when incorporated into FCI (max t= -3.210, max p = 1.47E-3); test-wise deletion also outperformed FCI with PMM but not by a significant margin (t = -1.331, p = 0.184). Test-wise deletion did however outperform all of the other 7 methods with RFCI (max t=-3.482, max p=5.73E-4). Moreover, test-wise deletion conserves an average of 8.96\% more samples per CI test (95\% CI: 7.95-9.98\%) than list-wise deletion for FCI and similarly 8.82\% (95\% CI: 7.82-9.82\%) for RFCI. On the other hand, heuristic test-wise deletion conserves only 1.05\% (95\% CI: 1.00-1.10\%) more samples than test-wise deletion for FCI and only 0.98\% (95\% CI: 0.94-1.02\%) more samples for RFCI. We conclude that the real data results largely replicate the synthetic data results for the MNAR case.

\section{Conclusion} \label{sec_conc}

We proposed test-wise deletion as a strategy to improve upon list-wise deletion for CCD algorithms even when MNAR holds. Test-wise deletion specifically involves running FCI or RFCI using Algorithm \ref{alg_wrapper} without pre-processing the missing values. We proved soundness of the procedure so long as the missingness mechanisms do not causally affect each other in the underlying causal graph. Moreover, experiments highlighted the superior sample efficiency of test-wise deletion as compared to list-wise deletion. We conclude that test-wise deletion is a viable alternative to list-wise deletion when MNAR holds. 

We ultimately hope that test-wise deletion will prove useful for investigators wishing to apply CCD algorithms on data with missing values. Test-wise deletion is easily implemented in a few lines of code via Algorithm \ref{alg_wrapper}. Here, we simply call a CCD algorithm equipped with Algorithm \ref{alg_wrapper} in place of a normal CI test.

\begin{acknowledgements}
Research reported in this publication was supported by grant U54HG008540 awarded by the National Human Genome Research Institute through funds provided by the trans-NIH Big Data to Knowledge initiative. The research was also supported by the National Library of Medicine of the National Institutes of Health under award numbers T15LM007059 and R01LM012095. The content is solely the responsibility of the authors and does not necessarily represent the official views of the National Institutes of Health.
\end{acknowledgements}

\bibliographystyle{spmpsci}
\bibliography{thesis_biblio}

\section{Appendix}

\subsection{Heuristic Test-Wise Deletion} \label{appendix_Htest}

We consider running FCI or RFCI with only line \ref{alg:first_query} of Algorithm \ref{alg_wrapper}; we therefore do not query the CI oracle with $\bm{S}_l$ when the CI oracle with $\bm{S}_{O_i O_j \bm{W}}$ outputs one. 

We refer to the above test-wise deletion strategy as \textit{heuristic test-wise deletion} because the procedure is not sound in general, even when Assumption \ref{assump1} holds. The problem lies in the inability to query the CI oracle with a consistent set of selection variables either directly (as with list-wise deletion) or indirectly (as with Algorithm \ref{alg_wrapper}). We thus often cannot soundly execute FCI or RFCI's orientation rules. For example, for FCI's R1, if we have the unshielded triple $O_i * \!\! \rightarrow O_k \circline O_j$ with $O_k \not \in \bm{An}(O_i, \bm{S}_{O_iO_j\bm{W}_1})$, (1) $O_i \ci_d O_j | (\bm{W}_2, \bm{S}_{O_iO_j\bm{W}_2})$ with $\bm{W}_2 \subseteq \bm{O} \setminus \{O_i, O_j\}$ minimal and (2) $O_k \in \bm{W}_2$, then we may claim that $O_k$ is an ancestor of $O_i, O_j$ or $\bm{S}_{O_iO_j\bm{W}_2}$ with (1) and (2) (but not $\bm{S}_{O_iO_j\bm{W}_1}$; see Lemma 14 in \cite{Spirtes99}). We thus cannot conclude in general that we have $O_k \in \bm{An}(O_j)$ by using the arrowhead at $O_k$; we can only conclude that $O_k \in \bm{An}(O_j, \bm{S}_{O_iO_j\bm{W}_2})$; this fact in turn prevents us from executing R1 by orienting $O_i * \!\! \rightarrow O_k \circline\!\! * O_j$ as $O_i * \!\! \rightarrow O_k \rightarrow O_j$.

We can however justify heuristic test-wise deletion under MCAR, where missing values do not depend on any other measured or missing values. One interpretation of MCAR in terms of a causal graph reads as follows:
\begin{assumption} \label{assump2}
There does not exist an undirected path between any member of $\bm{O}$ and any member of $\cup_{i=1}^q M_i$ in the underlying DAG.\footnote{This MCAR interpretation implies that $\cup_{i=1}^q M_i \ci_d \bm{O}$, so the interpretation is similar to the MCAR interpretation introduced in \cite{Mohan13}, where we have $\cup_{i=1}^q M_i \ci_d (\{\bm{O} \cup \bm{L} \cup \bm{S}\} \setminus \cup_{i=1}^q M_i)$.}
\end{assumption}

Now Assumption \ref{assump2} states that the set $\cup_{i=1}^q M_i$ plays no role in the conditional dependence relations between the observables. Specifically:
\begin{lemma} \label{lem_comb2}
Consider Assumption \ref{assump2}. Then $O_i \not \ci_d O_j |$ $(\bm{W}, \bm{S}_{O_iO_j\bm{W}})$ if and only if $O_i \not \ci_d O_j | (\bm{W}, \bm{S})$.
\end{lemma}
\begin{proof}
The proof follows trivially if $\bm{S}_{O_iO_j\bm{W}}=\bm{S}$, so assume that we have $\bm{S}_{O_iO_j\bm{W}}\supset \bm{S}$. Let $\bm{T}_{O_iO_j\bm{W}} = \{ \bm{S}_l \setminus \bm{S} \}$. Then no member of $\bm{T}_{O_iO_j\bm{W}}$ is on any undirected path between $O_i$ and $O_j$ by Assumption \ref{assump2}. Hence, no subset of $\bm{T}_{O_iO_j\bm{W}}$ can be used to block an active path $\pi$ between $O_i$ and $O_j$. This proves the backward direction. Moreover, no subset of $\bm{T}_{O_iO_j\bm{W}}$ can be used to activate any inactive path $\pi$ between $O_i$ and $O_j$. This proves the forward direction by contrapositive. \qed
\end{proof}
The corresponding statement to Theorem \ref{thm1} then reads as follows:
\begin{proposition} \label{thm2}
Consider Assumption \ref{assump2}. Further assume d-separation faithfulness. Then FCI using only line \ref{alg:first_query} of Algorithm \ref{alg_wrapper} outputs the same graph as FCI using a CI oracle with $\bm{S}$. The same result holds for RFCI.
\end{proposition}
\begin{proof}
It suffices to show that $O_i \not \ci O_j |$  $(\bm{W}, \bm{S})$ if and only if line \ref{alg:first_query} of Algorithm \ref{alg_wrapper} outputs zero. This follows directly by Lemma \ref{lem_comb2} and d-separation faithfulness. \qed
\end{proof}

Notice however that Assumption \ref{assump2} is much more difficult to justify in practice than Assumption \ref{assump1}. We therefore do not recommend FCI or RFCI with only line \ref{alg:first_query} of Algorithm \ref{alg_wrapper} in general, because these algorithms may not be sound when dealing with real data. 

Heuristic test-wise deletion can nonetheless perform very well in the finite sample size case even when Assumption \ref{assump2} is violated due to the extra boost in sample size provided by avoiding list-wise deletion altogether. We have summarized the simulation results in Figures \ref{fig_heur} and \ref{fig_heur_skel} in the MNAR case. Heuristic test-wise deletion outperforms test-wise deletion slightly by at most 0.203 SHD points on average (Figures \ref{fig_heur:FCI_SHD_heur} and \ref{fig_heur:RFCI_SHD_heur}). We could account for the increase in performance by a 5-15\% increase in the average sample size per CI test compared to test-wise deletion (Figures \ref{fig_heur:FCI_samples_heur} and \ref{fig_heur:RFCI_samples_heur}). However, heuristic test-wise deletion generally underperforms test-wise deletion in skeleton discovery by a margin gradually increasing in sample size. This dichotomy between the overall SHD and the skeleton SHD occurs because accurate endpoint orientation requires more samples than accurate skeleton discovery in general. We conclude that while heuristic test-wise deletion usually outperforms test-wise deletion when taking endpoint orientations into account, the performance improvement is small.

The results for the MAR case follow similarly, as summarized in Figure \ref{fig_heur_rank}. Heuristic test-wise deletion claims an average lower rank than test-wise deletion due to a 5-20\% increase in sample size in this scenario.

\begin{figure*}
\centering
\begin{subfigure}{0.4\textwidth}
  \centering
  \includegraphics[width=0.8\linewidth]{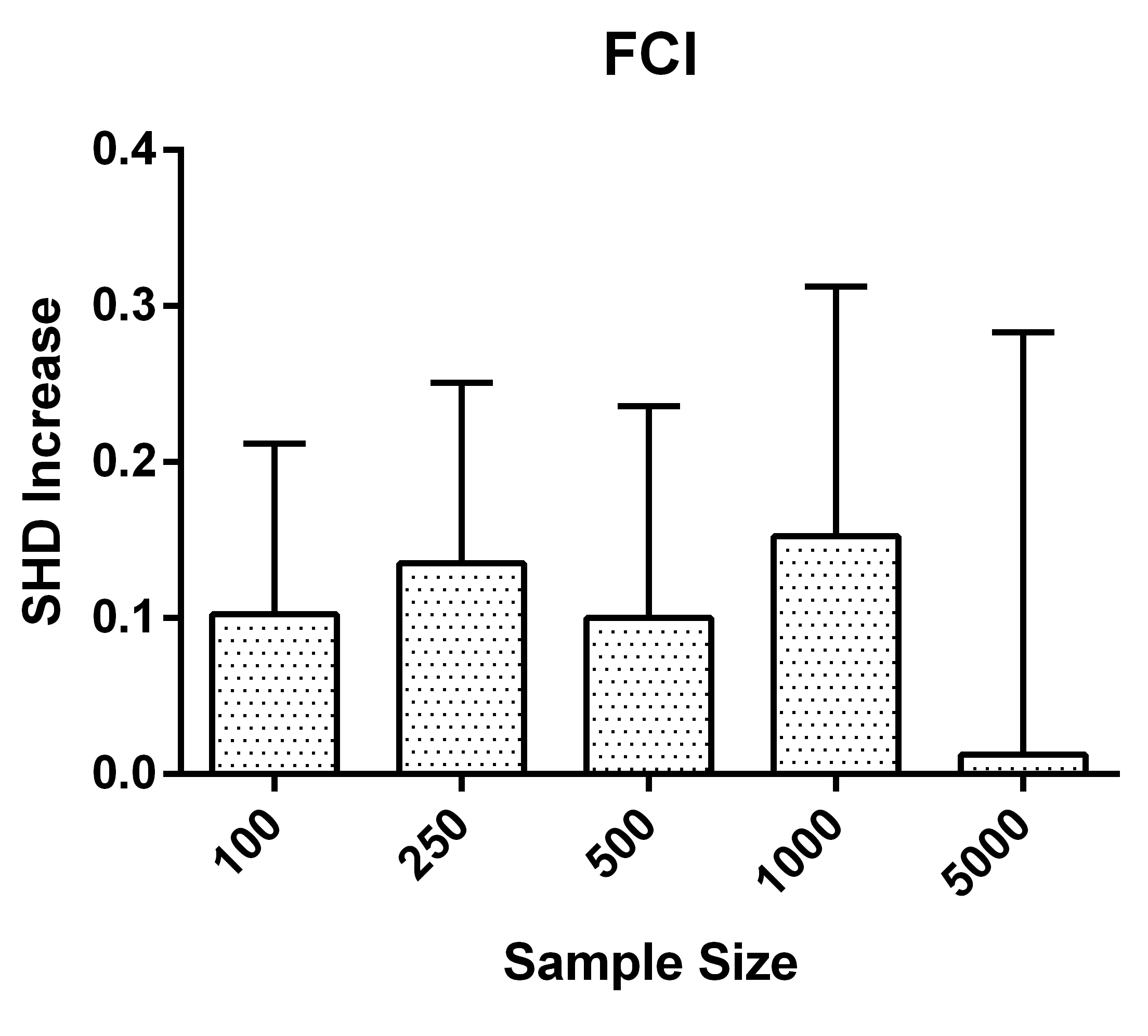}
  \caption{}
  \label{fig_heur:FCI_SHD_heur}
\end{subfigure}
\begin{subfigure}{0.4\textwidth}
  \centering
  \includegraphics[width=0.8\linewidth]{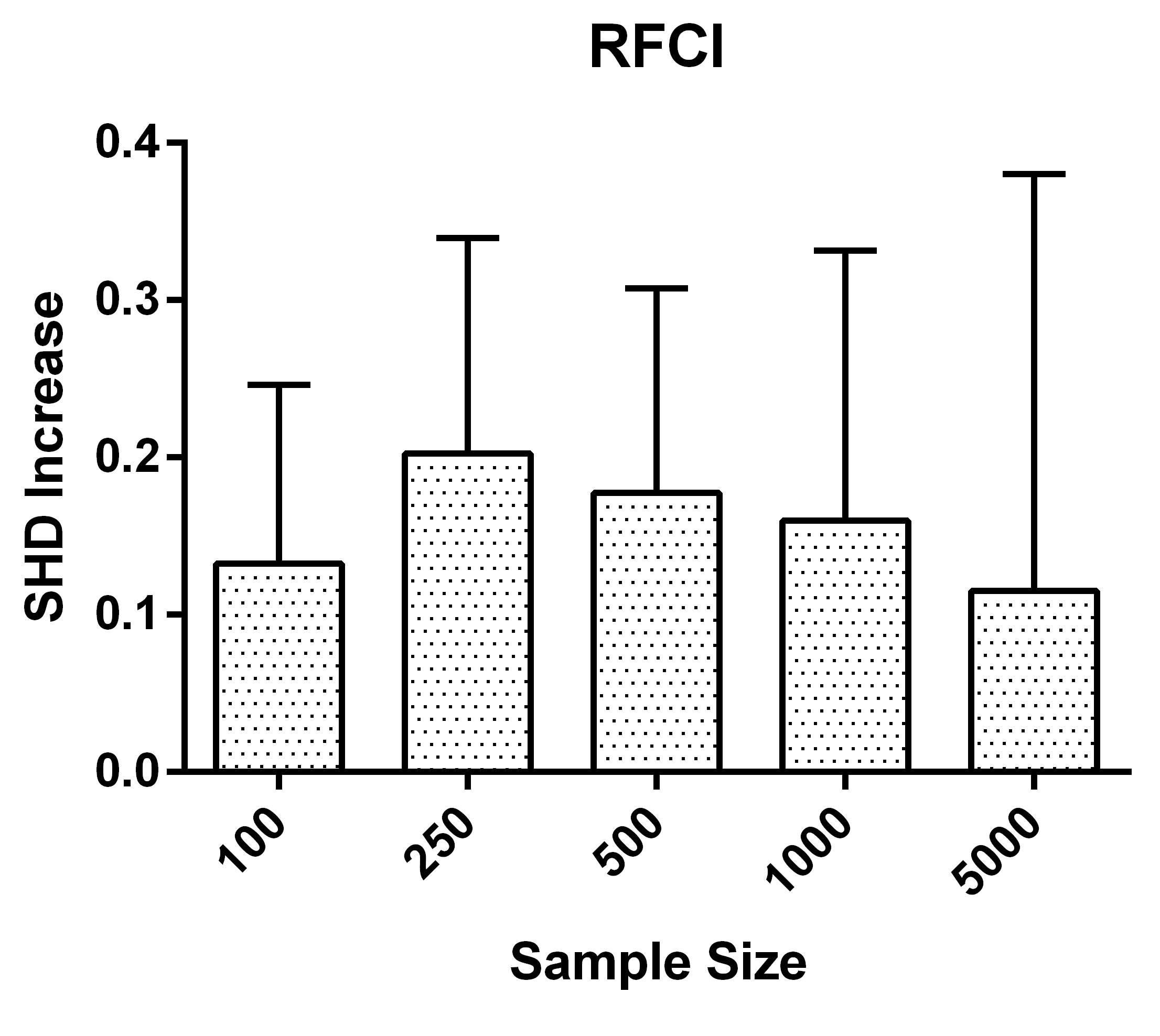}
  \caption{}
  \label{fig_heur:RFCI_SHD_heur}
  \end{subfigure}
  
  \begin{subfigure}{0.4\textwidth}
  \centering
  \includegraphics[width=0.8\linewidth]{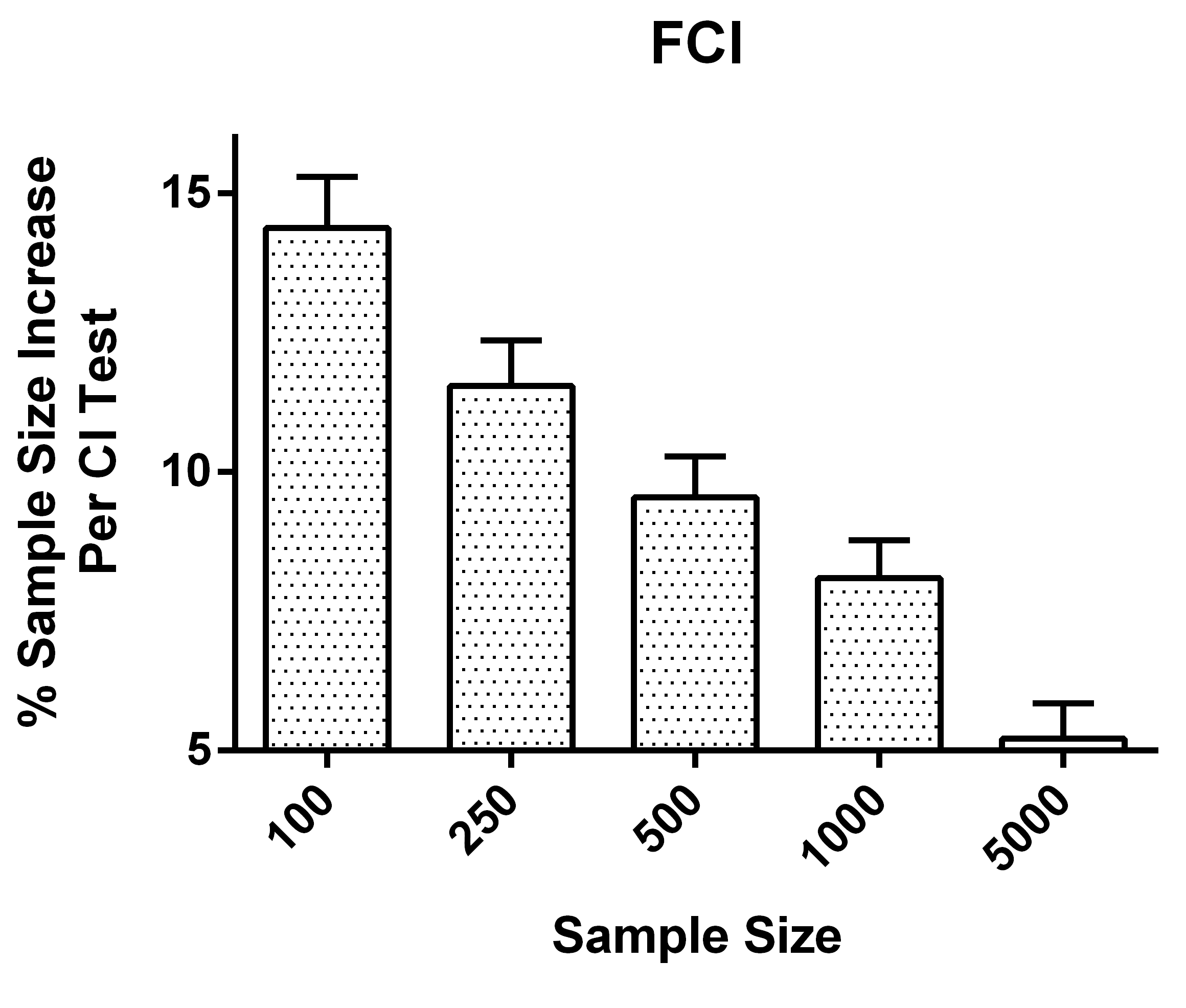}
  \caption{}
  \label{fig_heur:FCI_samples_heur}
\end{subfigure}
\begin{subfigure}{0.4\textwidth}
  \centering
  \includegraphics[width=0.8\linewidth]{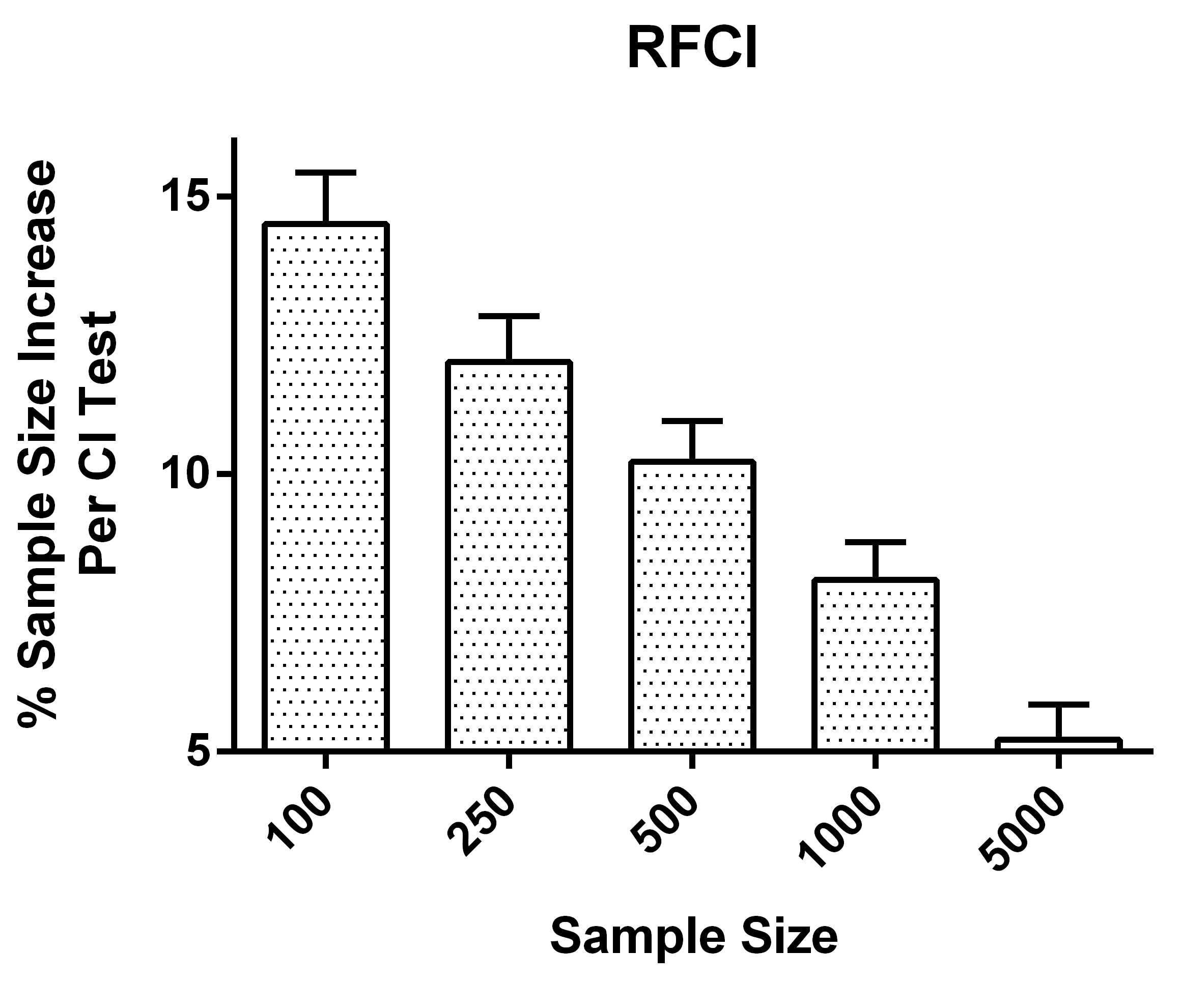}
  \caption{}
  \label{fig_heur:RFCI_samples_heur}
\end{subfigure}

\caption{Test-wise deletion vs. heuristic test-wise deletion in the MNAR case. We find that test-wise deletion underperforms heuristic test-wise deletion by yielding slightly larger SHD values on average according to (a) and (b); notice that the y-axis corresponds to an \textit{increase} in the SHD rather than a decrease. Subfigures (c) and (d) show the increase in average sample size per CI test for heuristic test-wise deletion as compared to test-wise deletion.} \label{fig_heur}
\end{figure*}

\begin{figure*}
\centering
\begin{subfigure}{0.4\textwidth}
  \centering
  \includegraphics[width=0.8\linewidth]{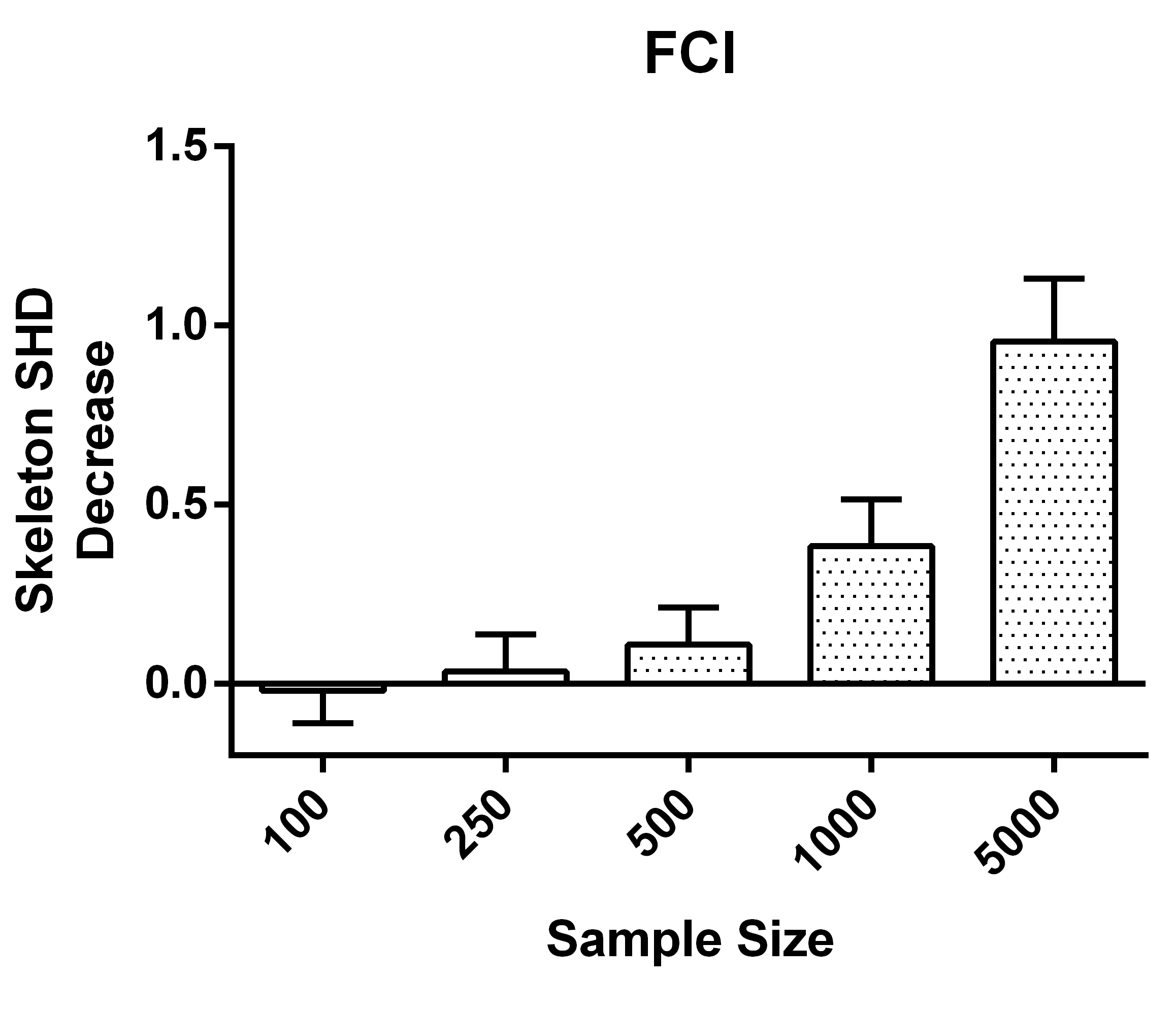}
  \caption{}
  \label{fig_heur_skel:FCI_SHD_heur_skel}
\end{subfigure}
\begin{subfigure}{0.4\textwidth}
  \centering
  \includegraphics[width=0.8\linewidth]{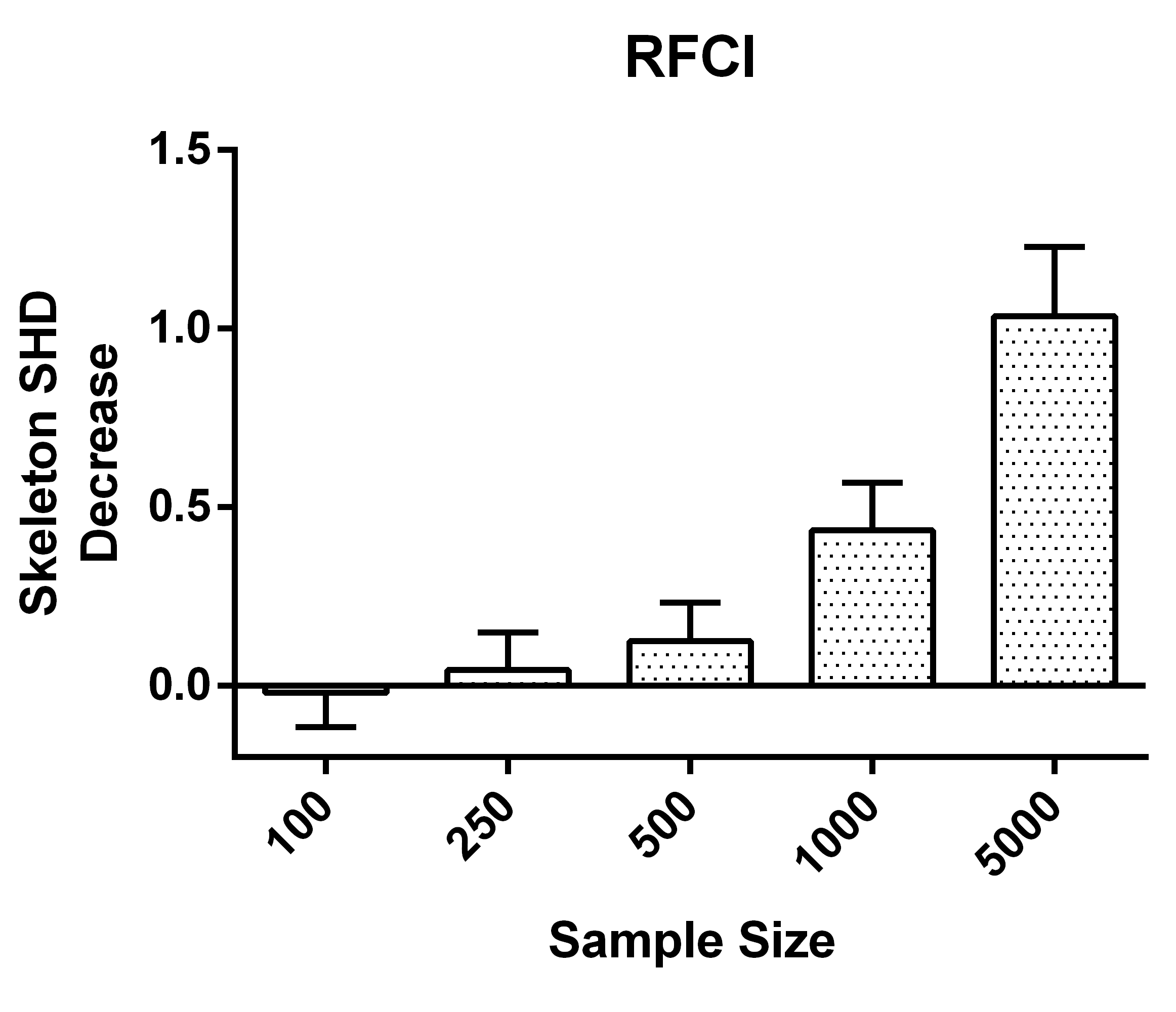}
  \caption{}
  \label{fig_heur_skel:RFCI_SHD_heur_skel}
\end{subfigure}

\caption{Test-wise deletion vs. heuristic test-wise deletion in skeleton discovery in the MNAR case. We find that test-wise deletion outperforms heuristic test-wise deletion by yielding smaller skeleton SHD values on average. Moreover, the margin gradually increases with sample size.} \label{fig_heur_skel}
\end{figure*}

\begin{figure*}
\centering
\begin{subfigure}{0.4\textwidth}
  \centering
  \includegraphics[width=0.8\linewidth]{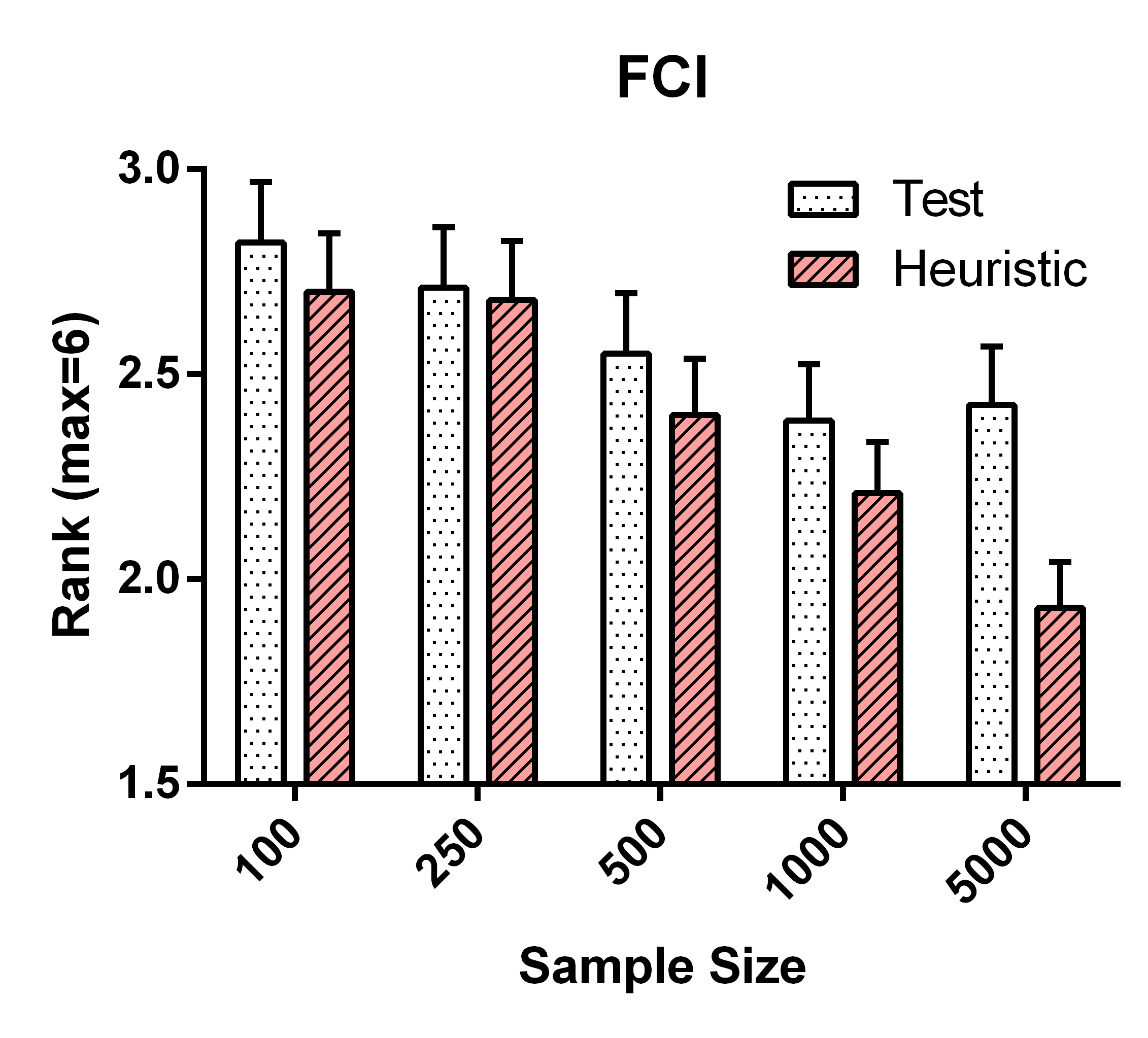}
  \caption{}
  \label{fig_heur_rank:FCI_heur_rank}
\end{subfigure}
\begin{subfigure}{0.4\textwidth}
  \centering
  \includegraphics[width=0.8\linewidth]{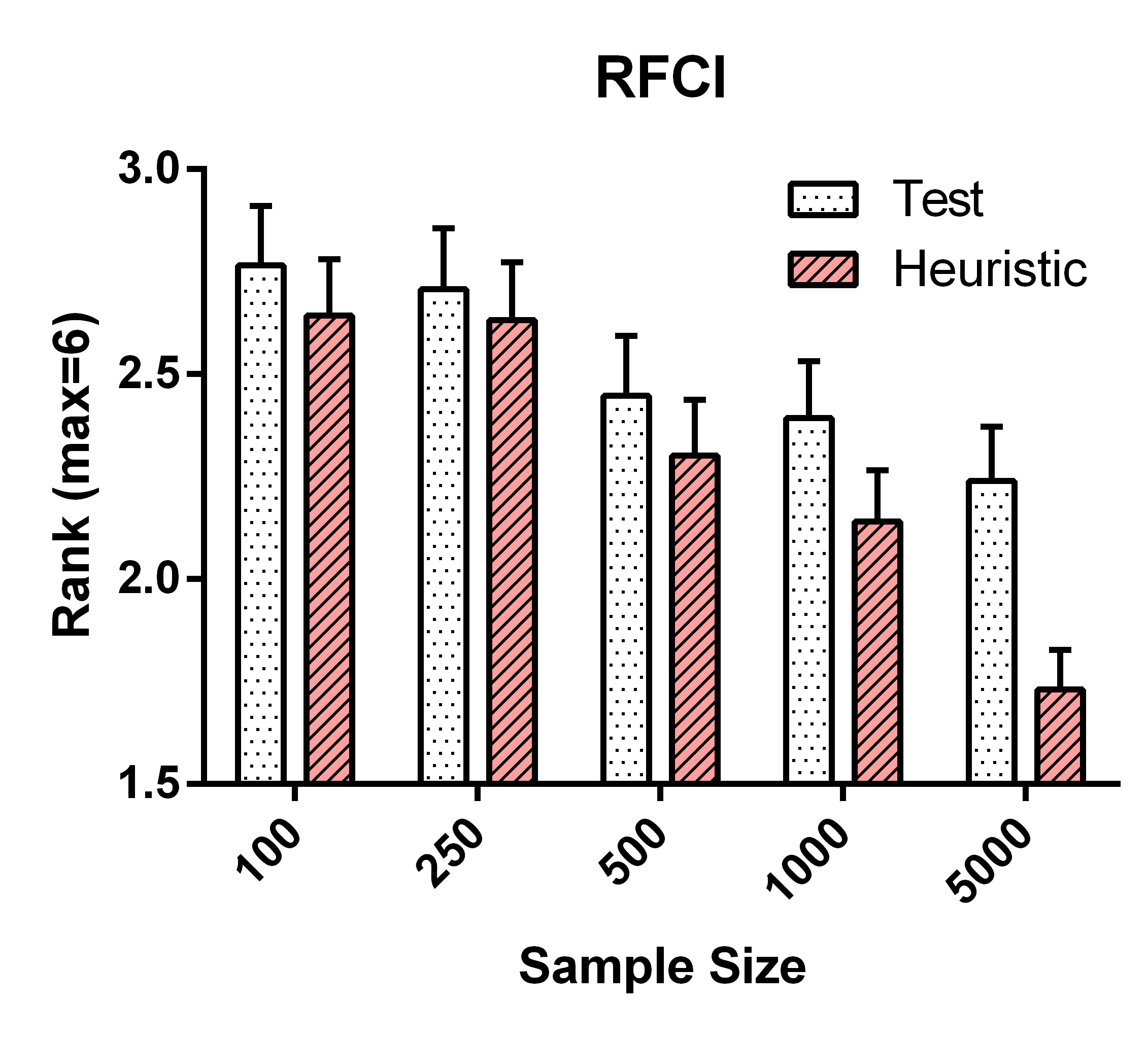}
  \caption{}
  \label{fig_heur_rank:RFCI_heur_rank}
  \end{subfigure}
  
  \begin{subfigure}{0.4\textwidth}
  \centering
  \includegraphics[width=0.8\linewidth]{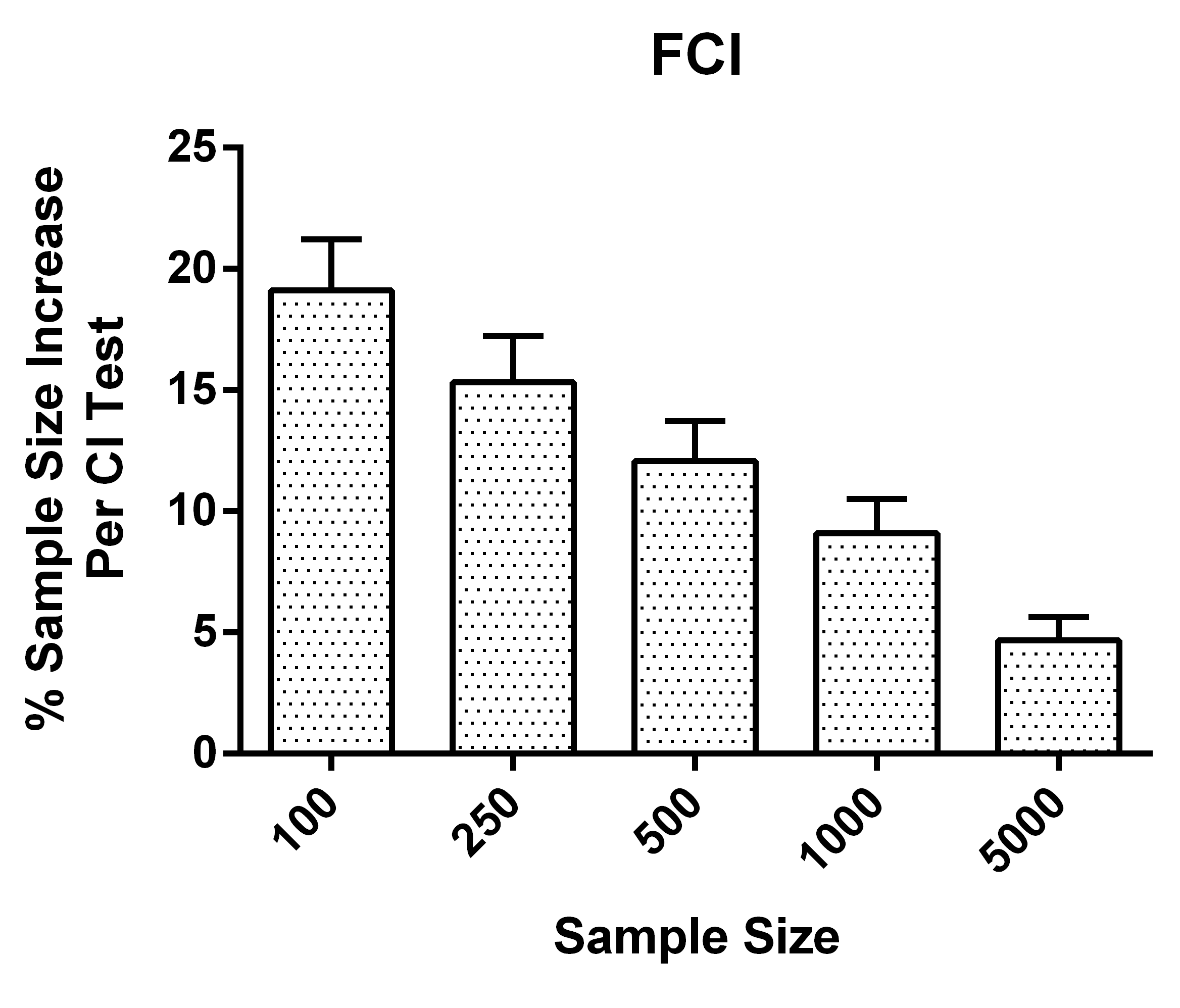}
  \caption{}
  \label{fig_heur_rank:FCI_samples_heur_MAR}
\end{subfigure}
\begin{subfigure}{0.4\textwidth}
  \centering
  \includegraphics[width=0.8\linewidth]{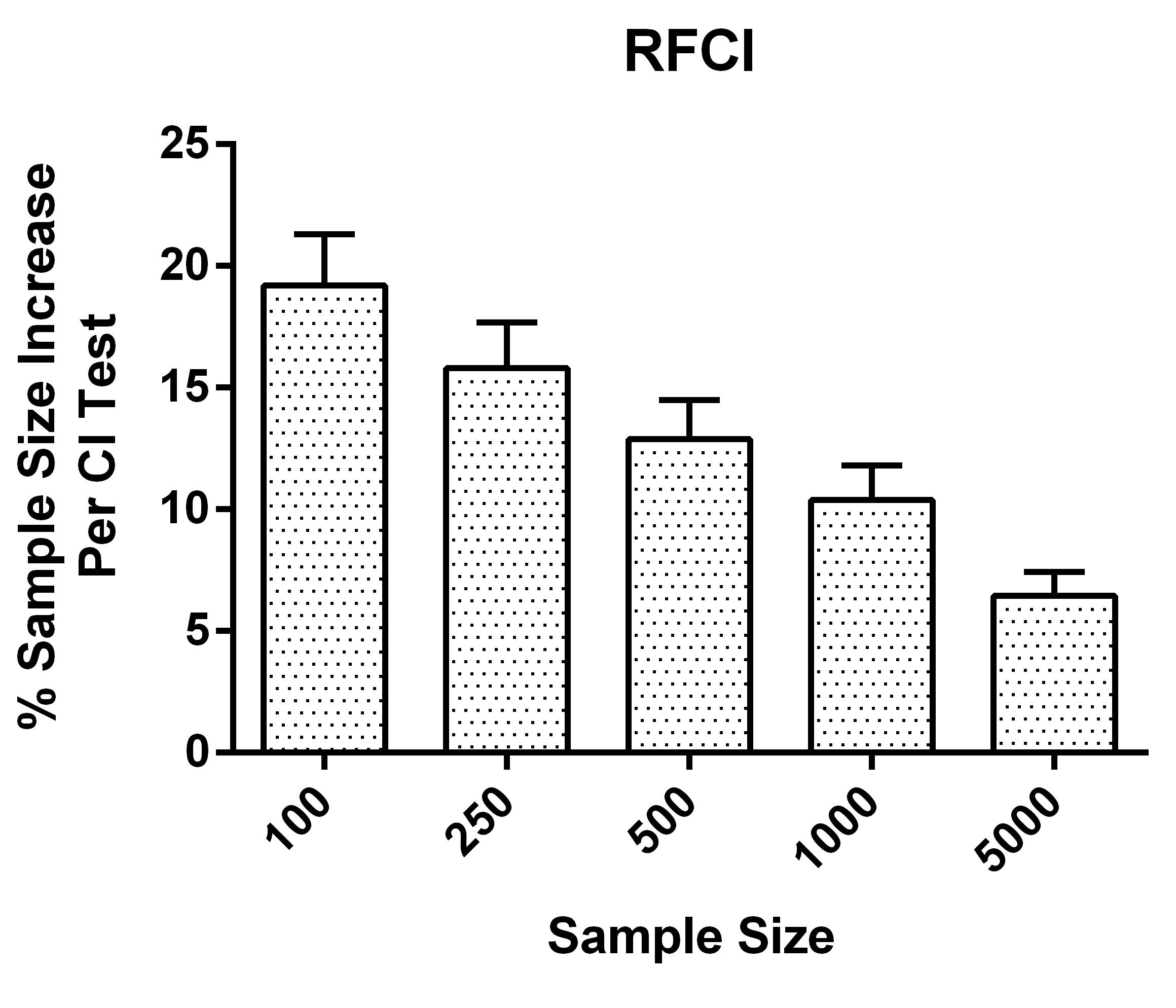}
  \caption{}
  \label{fig_heur_rank:RFCI_samples_heur_MAR}
\end{subfigure}

\caption{Test-wise deletion vs. heuristic test-wise deletion as compared to five imputation methods in the MAR case. Heuristic test-wise deletion has a smaller average rank than test-wise deletion for FCI in (a) and RFCI in (b). The performance increase of heuristic test-wise deletion results because of increased sample efficiency for FCI in (c) and RFCI in (d).} \label{fig_heur_rank}
\end{figure*}

\subsection{Test-Wise Deletion vs. Imputation} \label{appendix_tvi}

We have summarized the results of test-wise deletion vs. imputation for the MNAR case in Figure \ref{fig_imp}. Test-wise deletion outperforms all imputation methods by a large margin in this regime.

\begin{figure*}
\centering
\begin{subfigure}{0.4\textwidth}
  \centering
  \includegraphics[width=0.8\linewidth]{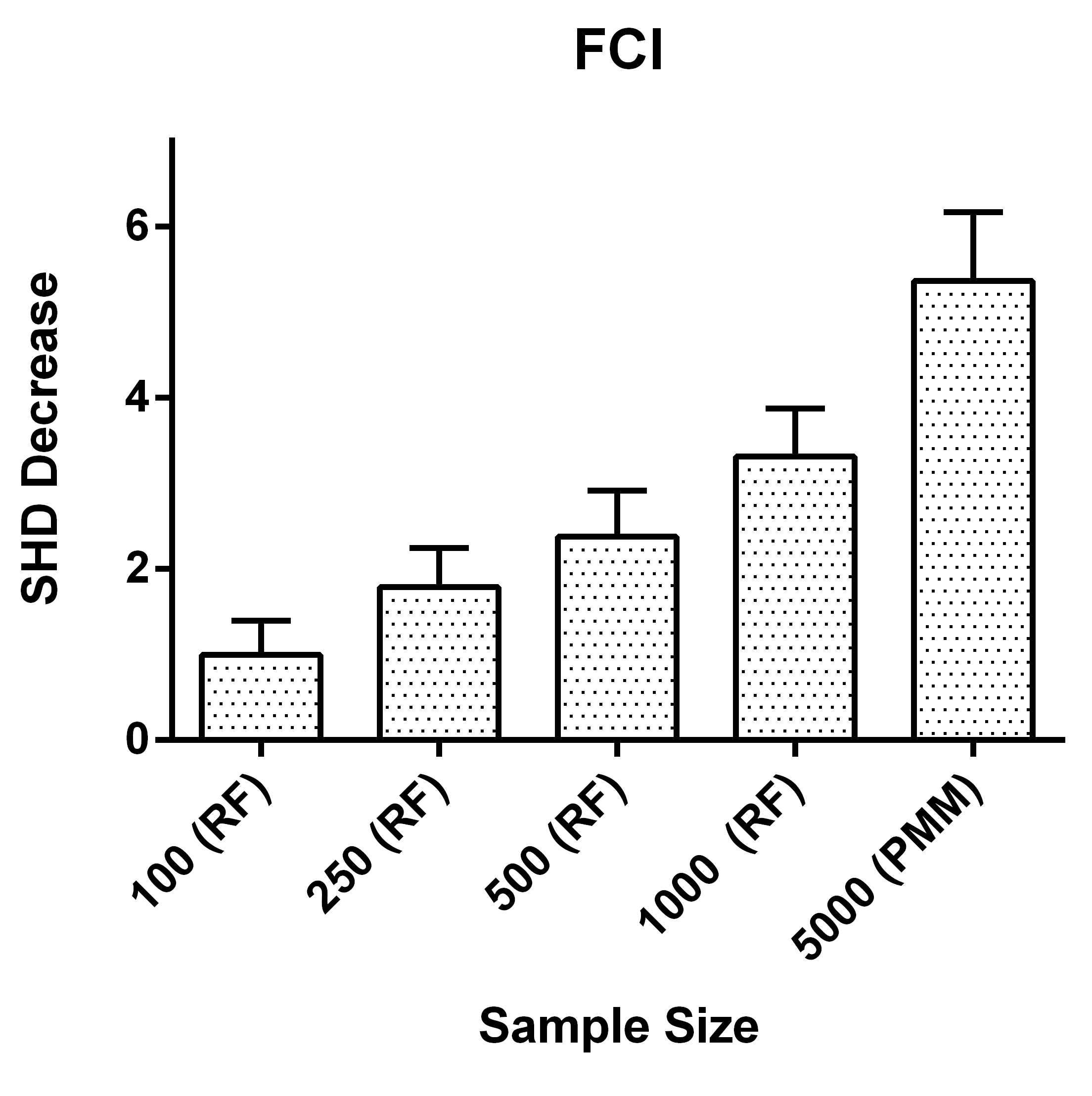}
  \caption{}
  \label{fig_imp:FCI_SHD_imp}
\end{subfigure}
\begin{subfigure}{0.4\textwidth}
  \centering
  \includegraphics[width=0.8\linewidth]{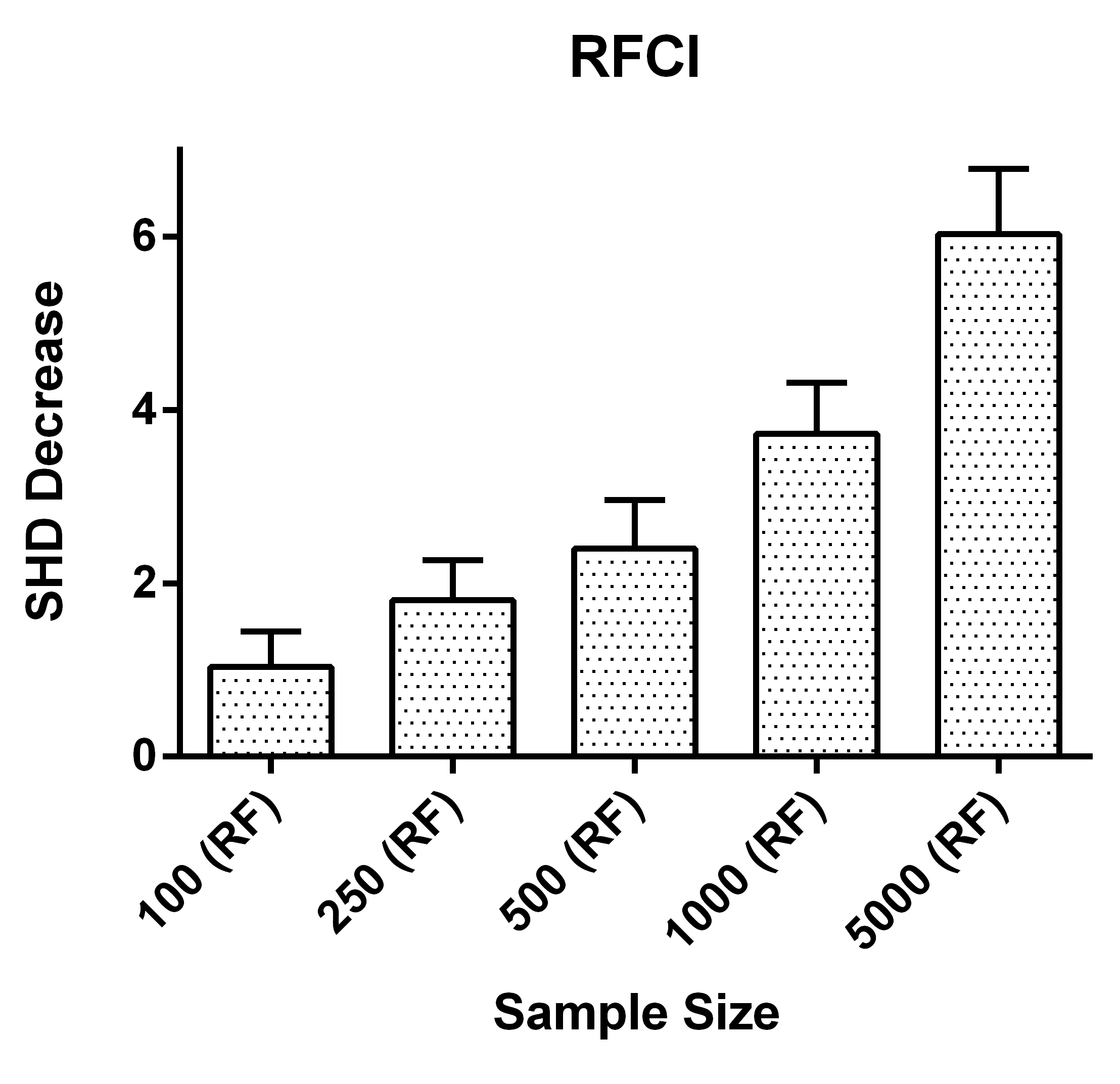}
  \caption{}
  \label{fig_imp:RFCI_SHD_imp}
\end{subfigure}

\caption{Performance of test-wise deletion vs. the best of five imputation methods in terms of the SHD when MNAR holds. Test-wise deletion outperforms the best imputation method (usually RF) by an increasing margin as sample size increases.} \label{fig_imp}
\end{figure*}

\end{document}